\documentclass[
  a4paper,onecolumn,
  accepted=2024-11-26,
]{quantumarticle}
\pdfoutput=1

\usepackage{stylesheet}

\begin{document}

\title{Projective toric designs, quantum state designs, and mutually unbiased bases}

\author{\iosuename}\orcid{0000-0003-3383-1946}
\email{jtiosue@umd.edu}
\affiliation{\QUICS}
\affiliation{\JQI}

\author{T.~C.~Mooney}\orcid{0000-0001-9727-6967}
\email{tmooney@umd.edu}
\affiliation{\QUICS}
\affiliation{\JQI}

\author{Adam~Ehrenberg}\orcid{0000-0002-3167-6519}
\affiliation{\QUICS}
\affiliation{\JQI}

\author{Alexey~V.~Gorshkov}\orcid{0000-0003-0509-3421}
\affiliation{\QUICS}
\affiliation{\JQI}

\date{November 26, 2024}

\begin{abstract}
    \vspace{-1.5em}
    Trigonometric cubature rules of degree $t$ are sets of points on the torus over which sums reproduce integrals of degree $t$ monomials over the full torus.
    They can be thought of as $t$-designs on the torus, or equivalently $t$-designs on the diagonal subgroup of the unitary group.
    Motivated by the projective structure of quantum mechanics, we develop the notion of $t$-designs on the \textit{projective} torus, which, surprisingly, have a much more restricted structure than their counterparts on full tori.
    We provide various new constructions of toric and projective toric designs and prove bounds on their size.
    We draw connections between projective toric designs and a diverse set of mathematical objects, including difference and Sidon sets from the field of additive combinatorics, symmetric, informationally complete positive operator valued measures (SIC-POVMs) and complete sets of mutually unbiased bases (MUBs) from quantum information theory, and crystal ball sequences of certain root lattices.
    Using these connections, we prove bounds on the maximal size of dense $B_t \bmod m$ sets.
    We also use projective toric designs to construct families of quantum state designs.
    In particular, we construct families of (uniformly-weighted) quantum state $2$-designs in dimension $d$ of size exactly $d(d+1)$ that do not form complete sets of MUBs, thereby disproving a conjecture concerning the relationship between designs and MUBs
    (Zhu 2015, \href{https://dx.doi.org/10.1103/PhysRevA.91.060301}{Phys.~Rev.~A 91, 060301}).
    We then propose a modification of Zhu's conjecture and discuss potential paths towards proving this conjecture.
    We prove a fundamental distinction between complete sets of MUBs in prime-power dimensions versus in dimension \( 6 \)
    (and, we conjecture, in all non-prime-power dimensions),
    the distinction relating to group structure of the corresponding projective toric design.
    Finally, we discuss many open questions about the properties of these projective toric designs and how they relate to other questions in number theory, geometry, and quantum information.
\end{abstract}

\maketitle

\vspace{-4em}

\tableofcontents

\section{Introduction}

Given a measure space $(M, \mu)$ and a set of polynomials on $M$, a \emph{$t$-design} on $M$ is a measure space $(X \subset M, \nu)$ satisfying
$\int_X f \dd \nu = \int_M f \dd \mu$
for all polynomials $f$ of degree $\leq t$ \cite{Gauss1866,Delsarte1977,Hardin1996,stroud_approximate_1971,beckersRelationCubatureFormulae1993,coolsMinimalCubatureFormulae1996,coolsConstructingCubatureFormulae1997,hammer_numerical_1956,baladram_on_2018,kuperbergNumericalCubatureArchimedes2004,kuperberg_numerical_ec_2004,victoir,seymour1984averaging}.
Classic examples are Gaussian quadrature rules \cite{Gauss1866} and spherical designs \cite{Delsarte1977,Hardin1996}, where the measure space $M$ is the hypercube and hypersphere, respectively. Typically, one is interested in finding designs where $X$ is a discrete measure space such that the integral over $X$ with respect to $\nu$ reduces to a weighted sum that is often simpler to compute. However, this is not always possible; in the case of rigged designs (defined below), it is often crucial that $X$ be a non-discrete measure space~\cite{iosueContinuousvariableQuantumState2022}.

Specific forms of $t$-designs for particular choices of measure spaces $M$ have found a plethora of uses in the field of quantum information theory~\cite{hoggarTDesignsProjectiveSpaces1982,hoggarParametersTDesignsFPd1984,bannaiTightDesignsCompact1985,wootters1989optimal,renesSymmetricInformationallyComplete2004,klappenecker2005mutually,dankert2005efficient, scott_tight_2006,ambainisQuantumTdesignsTwise2007,roberts2017chaos,Kueng2015,dankertExactApproximateUnitary2009,PhysRevLett.108.110503,aaronsonShadowTomographyQuantum2018,huangPredictingManyProperties2020,huangProvablyEfficientMachine2022,Acharya2021,kueng2016distinguishing,emersonScalableNoiseEstimation2005,knillRandomizedBenchmarkingQuantum2008,scalablemagesan,cross2016scalable,nielsenEntanglementFidelityQuantum1996,horodeckiGeneralTeleportationChannel1999,nielsenSimpleFormulaAverage2002,magesanGateFidelityFluctuations2011,luExperimentalEstimationAverage2015,Bravyi2021,ambainis2004small,hayden2004randomizing,kimmel2017phase,miInformationScramblingComputationally2021,sekino2008fast,hayden2007black,czartowski2020isoentangled-mu}.
In particular, complex projective space $\bbC\bbP^{d-1}$ describes the space of $d$-dimensional quantum states \cite{bengtsson_geometry_2008}, so $t$-designs on $M = \bbC\bbP^{d-1}$ are called \emph{complex-projective} or \emph{quantum state $t$-designs}. These quantum state designs also relate to other mathematical objects such as symmetric, informationally complete positive operator valued measures (SIC-POVMs) and complete sets of mutually unbiased bases (MUBs), which themselves are conjectured to relate to finite projective geometry. Finite-dimensional quantum state designs can be generalized to designs on infinite-dimensional, or continuous-variable, quantum systems by defining \emph{rigged quantum state $t$-designs}, which are designs on the space of tempered distributions $M=S(\bbR)'$~\cite{iosueContinuousvariableQuantumState2022}.
The (projective) unitary group $\mathrm{PU}(d)$ describes the space of noiseless dynamics of quantum states, and these too admit constructions of \emph{unitary $t$-designs}.
Finally, the space of mixed quantum states with the Hilbert-Schmidt volume measure allows for \textit{mixed-state $t$-designs} \cite{czartowski2020isoentangled-mu}.
Therefore, a better understanding of various kinds of $t$-designs can also lead to deep insights about quantum information.

Consider the complex sphere $\Omega_d$; that is, the set of unit vectors in $\bbC^d$. Any vector in $\Omega_d$ can be written (non-uniquely) as $\ket{q, \phi} \coloneqq \sum_{n=1}^d \sqrt{q_n} \e^{\i \phi_n} \ket n$, where $\{\ket{n}\}_{n=1}^{d}$ forms an orthonormal basis, $q = (q_{n})_{n=1}^{d}$ is a discrete probability distribution ($\sum_n q_n = 1$), and $\phi = (\phi_{n})_{n=1}^{d}$ is a set of phases. Therefore, $q$ belongs to the $(d-1)$-simplex $\Delta^{d-1}$ and $\phi$ to the $d$-torus $T^d$. Via this mapping $\Delta^{d-1}\times T^d \to \Omega_d$, one can combine simplex designs and toric designs to form complex spherical designs \cite{kuperbergNumericalCubatureArchimedes2004}. Identifying $\bbC\bbP^{d-1}$ with $\Omega_d / \U(1)$ (that is, quantum states are complex unit vectors with a global phase redundancy), we have a similar mapping $\Delta^{d-1} \times P(T^d) \to \bbC\bbP^{d-1}$ defined as $(q, [\phi]) \mapsto [\ket{q,\phi}]$, where $P(T^d) = T^d / \U(1)$ is the projective torus (see \cref{def:projective-torus}) and $[\cdot]$ denotes equivalence classes in the respective quotient spaces. In a similar way as before, via this mapping one can combine simplex designs and \textit{projective toric designs} (see \cref{def:projective-toric-design}) to form quantum state designs \cite{kuperbergNumericalCubatureArchimedes2004,iosueContinuousvariableQuantumState2022}.

In what follows, we flesh out and formalize this argument. Specifically, we formalize the notion of projective toric designs---both finite- and infinite-dimensional---and provide various constructions thereof. We discuss the connection between projective toric designs and difference sets \cite{taoAdditiveCombinatorics2006,obryant_complete_2004,bose_theorems_1962}, and use this correspondence to construct more projective toric designs, including some minimal ones.
We illustrate the connection to quantum state designs and various other mathematical objects.
Using minimal projective toric \( 2 \)-designs, we construct an infinite family of almost-minimal complex-projective \( 2 \)-designs.
Finally, we discuss many exciting open questions regarding projective toric designs, some of which are deeply connected to long-outstanding conjectures in mathematics, such as some conjectures relating to finite affine and projective spaces.
In particular, we construct explicit counterexamples to a conjecture by Zhu \cite[Conj.~1]{zhu2015mutually-unbias} on the relationship between uniformly-weighted quantum state \( 2 \)-designs and complete sets of MUBs, and we prove a fundamental distinction between complete sets of MUBs in prime-power dimensions versus in dimension \( 6 \).

\textit{Relation to prior work.}~
Toric designs have been considered before. Trigonometric cubature rules are such designs on the torus \cite{beckersRelationCubatureFormulae1993,coolsMinimalCubatureFormulae1996,coolsConstructingCubatureFormulae1997}. Ref.~\cite{kuperbergNumericalCubatureArchimedes2004} generalized the idea of trigonometric cubature to more general algebraic tori. Ref.~\cite{iosueContinuousvariableQuantumState2022} studied designs on projective tori and showed an equivalence to a specific case of Ref.~\cite{kuperbergNumericalCubatureArchimedes2004}, and further showed that such projective toric designs are related to complete sets of MUBs~\cite{durtMutuallyUnbiasedBases2010}. However we believe the presentation given in \cref{sec:torus} gives new clarity and focus on the subject. Furthermore, \cref{sec:constructions} compiles, to the best of our knowledge, all previously known constructions of projective toric $t$-designs\footnote{Of course, many toric designs are known, and these always project to projective toric designs. Such constructions are not compiled in this manuscript.}, and indeed generalizes some of these constructions.

The main novel contributions of our work lie in \cref{sec:minimal,sec:difference,sec:quantum}.
In \cref{sec:minimal}, we prove a general lower bound on the size of projective toric $t$-designs for all dimensions and all $t$ by relating these designs to the crystal ball sequence corresponding to the root lattice $A_{n-1}$ \cite{conway1997lowdimensional-,oeis-crystal-ball}.
In \cref{sec:difference}, we relate difference sets to projective toric designs. We show how the former can be used to construct the latter. Using the connection between difference sets and projective toric designs, we furthermore relate dense difference sets to the crystal ball sequence mentioned above, and derive new (to the best of our knowledge) bounds on the size of $B_t \bmod m$ sets (\textit{cf.}~\cref{cor:Bt-set-bound}).
In \cref{sec:singer-design-family}, using our construction of projective toric $t$-designs for all $t$ and dimensions $n$, we construct corresponding toric $t$-designs for $t$ and $n$ (where recall that a toric design is also a design on the diagonal subgroup of the unitary group).
In \cref{sec:quantum-from-toric}, we describe the relationship between projective toric designs and quantum state designs. This relationship was first noted in Refs.~\cite{kuperbergNumericalCubatureArchimedes2004,iosueContinuousvariableQuantumState2022}, though we believe that \cref{sec:quantum-from-toric} greatly clarifies the details of this connection. In \cref{sec:quantum-from-toric}, we also construct an infinite family of \textit{almost-minimal} quantum state $2$-designs---that is, quantum state $2$-designs of size exactly one more than minimal. While these specific almost-minimal designs have been noted before in Ref.~\cite{bodmannAchievingOrthoplexBound2015}, we arrive at the construction via a different route that utilizes projective toric designs, which we believe may have a better hope of generalizing to other infinite families and $t > 2$.

In \cref{sec:mubs-zhu}, we use projective toric \( 2 \)-designs from our difference set construction to yield uniformly-weighted quantum state \( 2 \)-designs in dimension \( d \) of size exactly \( d(d+1) \) that do \emph{not} form complete sets of MUBs, thereby disproving a conjecture by Zhu that has been open for nine years \cite[Conj.~1]{zhu2015mutually-unbias} (see also Ref.~\cite{avella2024cyclic-measurem} for a discussion on Zhu's conjecture).
In \cref{sec:mubs-group}, we further characterize the relationship between projective toric \( 2 \)-designs and complete sets of MUBs by proving (\emph{cf}.~\cref{prop:nongroup-dim-6}) that the phases involved in any complete set of MUBs in dimension \( 6 \) must form a non-group projective toric \( 2 \)-design.
In particular, this highlights a fundamental distinction between all known constructions of complete sets of MUBs in prime-power dimensions versus any potential construction in dimension \( 6 \) (and, we conjecture, in all non-prime-power dimensions).
We then discuss one possible modification of Zhu's conjecture relating to this fundamental distinction and discuss potential paths towards proving this new conjecture (\emph{cf.}~\cref{conj:mubs}).
Finally, \cref{sec:conclusion} compiles a number of new interesting open problems involving projective toric designs, highlighting their connection to a number of other open problems in mathematics.

\section{Theory of projective toric designs}
\label{sec:torus}

We begin with some basic definitions that are used throughout the rest of the paper.

\begin{definition}[Torus]
    Let $T \coloneqq \bbR / 2\pi\bbZ$. When $n\in\bbN$, let $I_n \coloneqq \set{1,2,\dots n}$; when $n=\infty$, let $I_n = I_\infty \coloneqq \bbN$. For such $n$, let $T^n \coloneqq \prod_{i\in I_n} T$ with the product topology. Define the projection maps $p_i\colon T^n \to T$ as $(\phi_j)_{j\in I_n} \mapsto \phi_i$. For all $n\in\mathbb{N}\cup\{\infty\},$ let $\mu_n$ denote $T^n$'s unit-normalized Haar measure.
\end{definition}

Note that by Tychonoff's theorem, $T^\infty$ is compact. For all $n$, $T^n$ is therefore a compact abelian group and thus has a unique unit-normalized Haar measure.

By definition, the product topology on $T^\infty$ is the coarsest topology such that the projection maps $p_i$ are continuous. Similarly, we endow $T^\infty$ with the smallest $\sigma$-algebra such that the projections $p_i$ are measurable. This $\sigma$-algebra is generated by sets of the form $A = \prod_{i\in\bbN} A_i$, where each $A_i$ is a measurable subset of $T$ and all but finitely many $A_i$ are equal to $T$. Define a measure $\mu'$ on $T^\infty$ by $\mu'(A) = \prod_{i\in\bbN} \mu_1(A_i)$. From Ref.~\cite[Thm.~10.6.1]{cohnMeasureTheory2013} (or Ref.~\cite{saekiProofExistenceInfinite1996} for a shorter proof), this definition of $\mu'$ on such subsets uniquely determines $\mu'$ on the whole space. Clearly $\mu'$ is transitionally-invariant and unit-normalized, and therefore $\mu' = \mu_\infty$.

We now define trigonometric cubature rules
\cite{beckersRelationCubatureFormulae1993,coolsMinimalCubatureFormulae1996,coolsConstructingCubatureFormulae1997}, which are designs on the torus.
To match the general terminology of this paper, we prefer to use the term \textit{toric design}.

\begin{definition}[Toric design]
    \label{def:toric-design}
    A \emph{$T^n$ $t$-design} (or trigonometric cubature rule of dimension $n$ and degree $t$ \cite{beckersRelationCubatureFormulae1993,coolsMinimalCubatureFormulae1996,coolsConstructingCubatureFormulae1997}) is a measure space $(X\subset T^n,\Sigma, \nu)$ such that
    \begin{equation}
        \int_X \exp\pargs{\i\sum_{j=1}^n \alpha_j \phi_j} \dd{\nu(\phi)}
        =
        \int_{T^n} \exp\pargs{\i\sum_{j=1}^n \alpha_j \phi_j} \dd{\mu_n(\phi)}
    \end{equation}
    for all $\alpha \in \bbZ^n$ satisfying $\sum_{j=1}^n \abs{\alpha_j} \leq t$.
\end{definition}

The torus $T^n$ is the same as the maximal torus of the unitary group $T(\U(n))$ \cite{hall2015lie-groups-lie-}, and indeed a $T^n$ design is a design on $T(\U(n))$ \cite{kuperbergNumericalCubatureArchimedes2004}.
Since $T(\U(n))$ is the diagonal subgroup of the unitary group $\U(n)$, we see that toric designs can equivalently be thought of as designs on the diagonal subgroup of $\U(n)$. Such designs are of interest in quantum information theory \cite{haferkamp2023on-the-moments-}.

\begin{example}
    In this example, we consider \( n=1 \) and \( t=2 \). Let \( X \) be the discrete, uniformly-weighted measure space \( X = \set{0, 2\pi/3, 4\pi/3} \subset T^1 \). Then, for every integer \( -2 \leq \alpha \leq 2 \),
    \begin{equation}
        \frac{1}{\abs{X}}\sum_{\phi\in X}\e^{\i \alpha \phi}
        = \frac{1}{2\pi} \int_0^{2\pi} \e^{\i\alpha \theta}\dd\theta
        = \begin{cases}
            1 & \text{if } \alpha = 0                 \\
            0 & \text{if } 1\leq \abs{\alpha} \leq 2.
        \end{cases}
    \end{equation}
    Hence, \( X \) is a \( T^1 \) \( 2 \)-design.
\end{example}

We now consider the projective torus, an important object in the study of quantum mechanics because it removes a global phase redundancy (see \cref{sec:quantum}).

\begin{definition}[Projective torus]
    \label{def:projective-torus}
    Let $P(T^n)$ denote the \emph{projective torus} $P(T^n) \coloneqq T^n / T$, where here $T$ denotes the inclusion $T \hookrightarrow T^n$ by $T \ni \theta \mapsto (\theta, \theta, \dots) \in T^n$.
\end{definition}

In other words, $P(T^n)$ is the set points in $T^n$ identified up to a constant additive factor.
Clearly, for any $f:T^n\to \mathbb{C}$ to descend to a well-defined function on $P(T^n)$ it must be constant on the cosets of the diagonal subgroup; in other words, it must satisfy $f(\e^{\i\phi_1 + \i\theta}, \e^{\i\phi_2 + \i\theta}, \dots) = f(\e^{\i\phi_1}, \e^{\i\phi_2}, \dots)$ for all $\theta\in T$. Hence, when studying designs on $P(T^n),$ we need only consider monomials on $T^n$ where the degree and conjugate degree are equal. A degree $t$ monomial on $P(T^n)$ therefore lifts to $\exp\pargs{\i \sum_{k=1}^t (\phi_{a_k} - \phi_{b_k})}$ for $a,b\in I_n^t$. We are thus now in a position to define a $P(T^n)$ $t$-design.

\begin{definition}[Projective toric design]
    \label{def:projective-toric-design}
    Fix an $n\in\bbN\cup \set*{\infty}$ and $t\in\bbN$. Let $X \subset P(T^n)$ and $(X, \Sigma, \nu)$ be a measure space. $X$ is called a \emph{$P(T^n)$ $t$-design} if for all $a, b \in I_n^t$,
    \begin{equation}\label{eq:proj-toric-design}
        \int_X \exp\pargs{\i \sum_{j=1}^{t} (\phi_{a_j} - \phi_{b_j})} \dd{\nu(\phi)} = \int_{P(T^n)} \exp\pargs{\i \sum_{j=1}^{t} (\phi_{a_j} - \phi_{b_j})} \dd{\mu_{n-1}(\phi)}.
    \end{equation}
    Here we denote the unit-normalized Haar measure on $P(T^n)$ as simply $\mu_{n-1}$ since $P(T^n) \cong T^{n-1}$.
    $X$ is called \emph{discrete} if $\nu$ is a counting measure, and is called \emph{finite} if it is discrete and $\abs{X} < \infty$.
    If $X$ is finite, then $\abs{X}$ is called the \emph{size} of $X$.
\end{definition}

We note that, in the language of Ref.~\cite{kuperbergNumericalCubatureArchimedes2004}, a $P(T^n)$ design is a design on the maximal torus of the projective unitary group $T(\operatorname{PU}(n))$. It was shown in Ref.~\cite{iosueContinuousvariableQuantumState2022} that the two notions coincide\footnote{We note that, in contrast to our manuscript, Ref.~\cite{iosueContinuousvariableQuantumState2022} refers to $P(T^n)$ designs as $T^n$ designs and refers to \( T^n \) designs as trigonometric cubature rules.}.
Clearly a $P(T^n)$ $t$-design is also a $(t-1)$-design, since we can let $a_t = b_t$ and have the integrand become an arbitrary degree $(t-1)$ monomial. Additionally, a $P(T^n)$ $t$-design is also a $P(T^{n-1})$ $t$-design, as can be seen by picking a subset of indices.

\begin{example}
    In this example, we consider \( n=2 \) and \( t=1 \).
    For any point in \( P(T^2) \), which is itself an equivalence class, we choose a representative of the equivalence class to be zero in the first entry of the tuple.
    In other words, since the equivalence relation $\sim$ that we quotient \( T^n \) by to get \( P(T^n) = \,^{T^n}\!/\!_{\sim} \) is \( (\vartheta,\varphi) \sim (\vartheta+\theta,\varphi+\theta)  \), we can always choose \( \theta \) such that the first entry in the tuple is \( 0 \).

    Let \( X \) be the discrete, uniformly-weighted measure space \( X = \set{(0,0), (0, 2\pi/3), (0, 4\pi/3)} \subset P(T^2) \). Then,
    \begin{equation}
        \frac{1}{\abs{X}}\sum_{\phi\in X}\e^{\i (\phi_a - \phi_b)}
        = \frac{1}{2\pi} \int_0^{2\pi} \e^{\i (\theta_a - \theta_b)}\dd\theta_2
        = \begin{cases}
            1 & \text{if } a=b       \\
            0 & \text{if } a \neq b.
        \end{cases}
    \end{equation}
    Note that, since we fix the first entry of any element of \( P(T^2) \) to be \( 0 \), \( \theta_1 = 0 \) and the Haar measure on \( P(T^2) \) is \( \dd\theta_2 \).
    Hence, we see that \( X \) is a \( P(T^2) \) \( 1 \)-design.
\end{example}

Throughout this work, we use double braces to denote multisets, whereas single braces denote sets as usual; that is, $\mset*{1,2,2} = \mset{2,1,2} \neq \mset{1,2}$, whereas $\set{1,2,2} = \set{1,2} = \set{2,1}$. Since the integrand in \cref{eq:proj-toric-design} contains only a finite number of projection maps, we can use Fubini's theorem to compute the integral on the right-hand side.
By choosing a set of representatives of $P(T^n)$ to be those phases $\phi$ for which $p_1(\phi) = \phi_1 = 0$, we can think of $P(T^n)$ as $\set*{0} \times T^{n-1}$. In this way, we have that $p_1(\phi) = 0$ for all $\phi$. It follows that $X \subset \set*{0}\times T^{n-1}$ is a $P(T^n)$ $t$-design if
\begin{salign}[eq:repr-proj-toric]
    \int_X\exp\pargs{\i \sum_{j=1}^t (\phi_{a_j} - \phi_{b_j})} \dd\nu(\phi)
    &= \int_{\set*{0}\times T^{n-1}} \exp\pargs{\i \sum_{j=1}^t (\phi_{a_j} - \phi_{b_j})} \dd{\mu_{n-1}(\phi)} \\
    &= \begin{cases}
        1 & \text{if } \mset{a_i \mid i\in \set*{1,\dots, t}} = \mset{b_i \mid i\in \set*{1,\dots, t}} \\
        0 & \text{otherwise}
    \end{cases}.
\end{salign}

Suppose that we set each $b_j = 1$. It follows that $X$ must match integration of polynomials on $T^{n-1}$ of degree $t$ and conjugate degree $0$ (because $\phi_{b_j} = 0$). Similarly, we can set each $a_j=1$, and thus $X$ must match integration of degree $0$ and conjugate degree $t$. More generally, we see that it must match on monomials on $T^{n-1}$ of degree $(t_1, t_2)$ whenever $t_1\leq t$ and $t_2 \leq t$. It follows that a $T^{n-1}$ $(2t)$-design is a $P(T^n)$ $t$-design, and a $P(T^n)$ $t$-design is a $T^{n-1}$ $t$-design. The reverse implications however do not hold in general.

By linearity, a $P(T^n)$ $t$-design exactly integrates all polynomials on $P(T^n)$ of degree $t$ or less.
It is the projective nature of the polynomials that we are integrating that give projective toric designs their interesting structure that is quite different than the structure of toric designs. For example, as we will see, for finite $n$, $P(T^n)$ $2$-designs must be of size at least $n(n-1)+1$, and indeed this can be saturated for many $n$; in contrast, it is known that a $T^n$ $4$-design requires size at least $2n^2$, $3$-design requires at least $4n$ points (which can often be achieved), and $2$-design requires at least $2n$ points (and $2n+1$ can often be achieved) \cite{coolsMinimalCubatureFormulae1996}.
Indeed, the difference between toric designs (\textit{i.e.}~trigonometric cubature rules) and projective toric designs is analogous to the difference between (complex) spherical designs and (complex) projective designs.

\subsection{Constructions of projective toric designs}
\label{sec:constructions}

In this section, we present a few simple constructions in order to get a handle on projective toric designs.
Later, in \cref{sec:difference}, we construct more (and smaller) projective toric $t$-designs by utilizing difference sets and Sidon sets from additive combinatorics \cite{taoAdditiveCombinatorics2006}.
Throughout this section, we write points in $P(T^n)$ as representatives in $T^n$ with the first entry set to $0$.

Our first example is a $P(T^n)$ $2$-design of size $n^2$ whenever $n$ is prime, and slightly larger when $n$ is not prime. Note that this construction can be generalized to be size $n^2$ whenever $n$ is a prime power, but we do not do this here. The generalized construction can be seen in the phases in the complete set of MUBs in prime-power dimensions given in Ref.~\cite{wootters1989optimal}.

\begin{theorem}[Thm.~C.9 of \cite{iosueContinuousvariableQuantumState2022}]
    \label{thm:primesconst}
    Let $n\in\bbN$.
    Define $p$ to be the smallest prime number strictly larger than $\max(2, n)$ (by the prime number theorem, $p \in \bigO{n + \log n}$). Let $X\subset T^n$ be the set
    \begin{equation}
        X = \set{\parentheses{0,2\pi (q_1+q_2)/p,2\pi (2q_1+4q_2)/p, \dots, 2\pi ((n-1)q_1+(n-1)^2 q_2)/p  } \mid  q_1 \in \bbZ_{p}, q_2 \in \bbZ_{p}}
    \end{equation}
    and $v$ the constant map\footnote{We corrected a minor error in Thm.~C.9 of \cite{iosueContinuousvariableQuantumState2022}. Namely, the map $v$ was stated as $v(\phi) = 1/p^2$. This is correct for all $n>2$, as $\abs{X} = p^2$. However, when $n=2$, $\abs{X} = p = 3$.} $v(\phi) = 1/\abs{X}$. Then $X$ with the counting measure weighted by $v$ is a $P(T^n)$ $2$-design.
\end{theorem}

We can easily write out the construction for $n=2$, where we have $p=3$, and therefore
$X = \{\left(0,0\right),\allowbreak\left(0,2\pi/3\right),\allowbreak\left(0, 4\pi/3\right)\}$ with weight $v(\phi)=1/3$ is a $P(T^2)$ $2$-design.
We show the construction in \cref{fig:design-example} for this example of $n=2$ with $p=3$ as well as for $n=3$ with $p=5$.

\begin{figure}
    \centering
    \begin{minipage}{.34\textwidth}
        \centerline{\includegraphics[width=.7\textwidth]{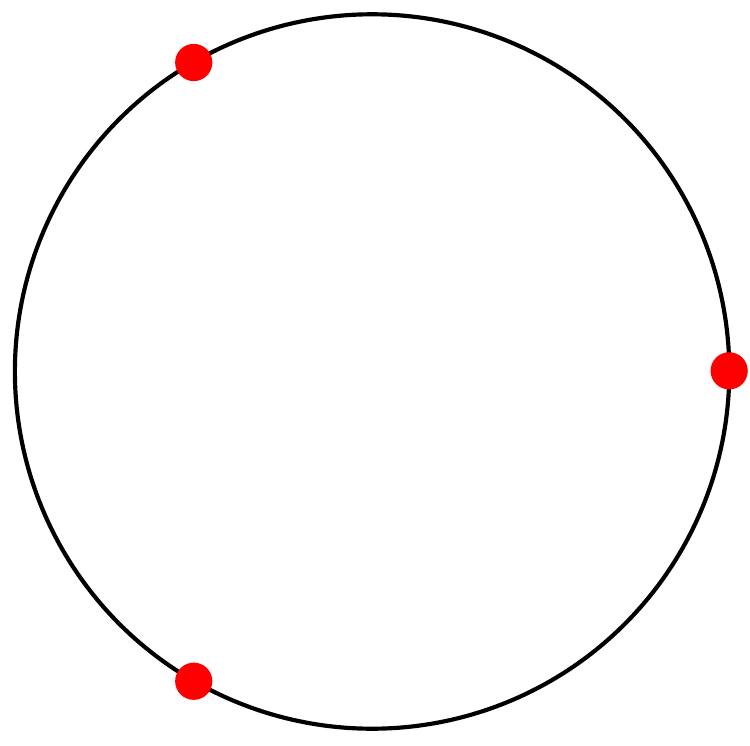}}
    \end{minipage}
    \hfill
    \begin{minipage}{.65\textwidth}
        \centerline{\includegraphics[width=.7\textwidth]{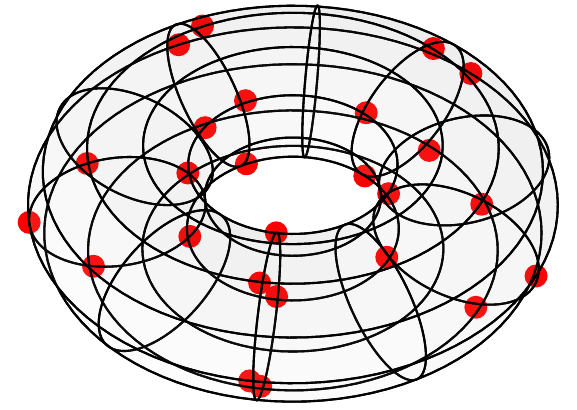}}
    \end{minipage}
    \caption{The construction of the $2$-design in \cref{thm:primesconst} for (left) $n=2$ with $p=3$ and (right) $n=3$ with $p=5$. Note we are representing points in $P(T^n)$ here as points in $T^{n-1}$ by discarding the first coordinate which we fix to $0$. The number of points in the design for (left) $n=2$ is $p$ and for (right) $n=3$ is $p^2 = 25$.}
    \label{fig:design-example}
\end{figure}

We can extend this construction to the case when $n=\infty$.

\begin{theorem}
    Let $X\subset T^\infty$ be the set
    \begin{equation}
        X = \set{\parentheses{0,\vartheta+\varphi, 2\vartheta+4\varphi, \dots, j\vartheta+j^2 \varphi,\dots  } \mid  \vartheta,\varphi \in [0,2\pi]}
    \end{equation}
    and $\nu$ the unit normalized Lebesgue measure on $[0,2\pi]^2$ (i.e.~$\dd\nu = \dd\vartheta\dd\varphi / (2\pi)^2$). Then $X$ is a $T^\infty$ $2$-design.
\end{theorem}
\begin{proof}
    For any $a,b,c,d\in\bbN$,
    \begin{salign}
        \int_X \exp\pargs{\i (\phi_a + \phi_b - \phi_c - \phi_d)} \dd\nu(\phi)
        &= \int_{[0,2\pi]^2} \exp\pargs{\i \vartheta (a+b-c-d) + \i\varphi (a^2+b^2-c^2-d^2)}  \frac{\dd\vartheta \dd\varphi}{(2\pi)^2}\\
        &= \begin{cases}
            1 & \text{if } a+b=c+d ~\land~ a^2+b^2=c^2+d^2 \\
            0 & \text{otherwise}
        \end{cases}\\
        &= \begin{cases}
            1 & \text{if } \mset{a,b} = \mset{c,d} \\
            0 & \text{otherwise}
        \end{cases},
    \end{salign}
    where in the last line we used \cite[Lem.~C.10]{iosueContinuousvariableQuantumState2022}.
\end{proof}

We now consider a construction for arbitrary $t$.

\begin{theorem}[Thm.~C.8 of \cite{iosueContinuousvariableQuantumState2022}]
    \label{thm:product-of-S1-designs}
    Let $n,t\in\bbN,$ and $X \subset T^n$ be the set
    \begin{equation}
        X = \set{\parentheses{0, 2\pi d_1/(t+1), 2\pi d_2/(t+1), \dots, 2\pi d_{n-1}/(t+1)} \mid d \in \bbZ_{t+1}^{n-1}},
    \end{equation}
    and $v$ be the constant map $v(\phi) = (t+1)^{-(n-1)}$. Then $X$ with the counting measure weighted by $v$ is a $P(T^n)$ $t$-design.
\end{theorem}

\begin{example}[$n=2$, $t=3$]
    We have
    \begin{equation}
        \begin{split}
            X = \Bigg\{
             & \left(0,0,0\right),\left(0,0,2\pi \frac{1}{3}\right), \left(0,2\pi \frac{1}{3},0\right), \left(0,0,2\pi \frac{2}{3}\right),\left(0,2\pi \frac{2}{3},0\right),                                                \\
             & \left(0,2\pi \frac{1}{3},2\pi \frac{2}{3}\right),\left(0,2\pi \frac{2}{3},2\pi \frac{1}{3}\right),\left(0,2\pi \frac{1}{3},2\pi \frac{1}{3}\right),\left(0,2\pi \frac{2}{3},2\pi \frac{2}{3}\right) \Bigg\},
        \end{split}
    \end{equation}
    with $v(\phi)=1/9,$ is a $P(T^3)$ $2$-design.
\end{example}

We now extend this construction to $n=\infty$.

\begin{theorem}
    Let $t\in\bbN$ and $X_1 \subset T$ be the discrete probability space $X_1 = \{2\pi d / (t+1) \mid \allowbreak d\in\bbZ_{t+1} \}$. Let $X = \prod_{i\in\bbN} X_1$ and its $\sigma$-algebra be generated by sets of the form $\prod_{i\in\bbN} A_i$ where each $A_i$ in the power set $A_i \in \calP(X_1)$ and for all but finitely many $i$ we have $A_i = X_1$. Define $\nu$ by its action $\nu(A) = \prod_{i\in\bbN} (\abs{A_i} / \abs{X_1})$, and note that $\nu$ uniquely extends to a measure on $X$ \cite[Thm.~10.6.1]{cohnMeasureTheory2013}. Then $X$ is a $T^\infty$ $t$-design.
\end{theorem}
\begin{proof}
    Let $m = \max\pargs{\max_j a_j, \max_j b_j}$.
    Since $t$ is finite, we are only ever dealing with a finite number of projection maps $p_i$ in the integrand. Therefore, we can apply Fubini's theorem to separate the integral $\int_X$ into a product of an integral over $X_1^m$ and an integral over the rest of the space.
    Hence,
    \begin{salign}
        \int_X \exp\pargs{\i \sum_{j=1}^t (\phi_{a_j} - \phi_{b_j})} \dd\nu(\phi)
        &= \frac{1}{\abs{\bbZ_{t+1}^m}}\sum_{d\in\bbZ_{t+1}^m} \exp\pargs{\frac{2\pi\i}{t+1} \sum_{j=1}^t (d_{a_j} - d_{b_j})}\\
        &= \begin{cases}
            1 & \text{if } \mset{a_j \mid j\in\set*{1,\dots, t}} = \mset{b_j \mid j\in\set*{1,\dots, t}} \\
            0 & \text{otherwise}
        \end{cases}.
    \end{salign}
\end{proof}

Finally, for completeness, we note the asymptotic existence theorem proven in Ref.~\cite{kuperbergNumericalCubatureArchimedes2004}.

\begin{theorem}[Thm.~3.3 and Cor.~5.4 of \cite{kuperbergNumericalCubatureArchimedes2004}]\label{thm:asymptoticdesigns}
    Asymptotically in $n\to\infty$ but for finite $n$, a $P(T^n)$ $t$-design must have size at least $\frac{n^t (1-o(1))}{\ceil{t/2}! \floor{t/2}!}$ and there exists $t$-designs of size $n^t(1+\littleo{1})$.
\end{theorem}

\subsection{Minimal projective toric designs}
\label{sec:minimal}

A very natural question that one can ask is \textit{what is the size of the smallest projective toric $t$-design?} We call such designs \textit{minimal}.
Ref.~\cite[Prop.~C.11]{iosueContinuousvariableQuantumState2022} proved a lower bound on the size of minimial projective toric $2$-designs.
In this section, we generalize this bound and prove a lower bound on the size of minimal projective toric $t$-designs for all $t$.
In the case when $t$ is even, we conjecture that this bound is tight.
In \cref{sec:difference}, we show that the $t=2$ bound can be saturated in many dimensions.

We begin by, for all $n,s\in\mathbb{N},$ defining the set
\begin{equation}
    \label{eq:Pst}
    P_s^{(n)} \coloneqq \set{\mathbf q - \mathbf r ~\bigg\vert~ \mathbf q, \mathbf r \in \bbN_0^n, ~ \sum_{i=1}^n q_i = \sum_{i=1}^n r_i = s}.
\end{equation}
An element $\mathbf{q} - \mathbf{r} \in P_s^{(n)}$ corresponds to a monomial $\exp\pargs*{\i \sum_{j=1}^n (q_j - r_j) \phi_j}$ on $P(T^n)$.
We show that $\abs*{P_s^{(n)}}$ is the $s^{\rm th}$ element of the crystal ball sequence corresponding to the root lattice $A_{n-1} \coloneqq \allowbreak \{\bm v \in \mathbb{Z}^n \mid \sum_{i=1}^n v_i =0\}$ \cite{conway1997lowdimensional-,oeis-crystal-ball}, and therefore arrive at the explicit formula for $\abs*{P_t^{(n)}}$ given in \cref{eq:Pst-size}. We begin by defining the crystal ball sequence of $A_{n-1}$. Let $S_{n-1}(t)$ denote the number of vertices of $A_{n-1}$ a distance $t$ away from some fixed vertex, where we define distance for the lattice $A_{n-1}$ as follows: letting $\mathcal{R}\coloneqq\{\mathbf{e}_i-\mathbf{e}_j\ \vert\ i,j\in\{1,\dots,n\}\}$ be the roots of $A_{n-1},$ the distance between $\mathbf{x},\mathbf{y}\in A_{n-1}$ is the smallest $d$ such that $\mathbf{x}-\mathbf{y}\in d\mathcal{R},$ where $d\mathcal{R}:= \mathcal{R}+\mathcal{R}+\dots+ \mathcal{R}$ is the $d$-fold set sum of $\mathcal{R}.$ The sequence $(S_{n-1}(t))_{t\in \bbN_0}$ is the \textit{coordination sequence} of $A_{n-1}$ \cite{conway1997lowdimensional-}. The \textit{crystal ball numbers} are the partial sums
$G_{n-1}(s) = \sum_{x=0}^s S_{n-1}(x)$ \cite{conway1997lowdimensional-}.
The explicit formula for $G_{n-1}(s)$ is \cite{conway1997lowdimensional-,oeis-crystal-ball}
\begin{equation}
    \label{eq:Pst-size}
    G_{n-1}(s) = \, _3F_2(1-n,-s,n;1,1;1)
    = \sum _{i=0}^s \binom{n-1}{i}^2 \binom{n-i+s-1}{s-i},
\end{equation}
where $\, _3F_2$ denotes the generalized hypergeometric function \cite{bailey_generalized_1964,slater_generalized_1966,petkovsek1996,zudilin_hypergeometric_2019}.
We can easily see that $P_s^{(n)} = s \calR$, and furthermore $G_{n-1}(s) = \abs{s\calR}$ by definition since it is precisely the set of all points that are reachable within a path of at most $s$ edges. It therefore follows that
\begin{equation}
    \label{eq:crystal-ball}
    \abs*{P_s^{(n)}} = G_{n-1}(s).
\end{equation}

We recall that Ref.~\cite{iosueContinuousvariableQuantumState2022} showed the equivalence of $P(T^n)$ designs and designs on the algebraic torus $T(\operatorname{PSU}(n))$ as defined in Ref.~\cite{kuperbergNumericalCubatureArchimedes2004}. Ref.~\cite{kuperbergNumericalCubatureArchimedes2004} further explored the connection between such designs and the root lattice of $\operatorname{PSU}(n)$, which is $A_{n-1}$.
This gives a hint as to why $A_{n-1}$ shows up in the analysis of projective toric designs.
Indeed, each point in $A_{n-1}$ corresponds to a monomial on $P(T^n)$. $P_s^{(n)}$ is precisely all points on $A_{n-1}$ a distance of less than or equal to $s$ from the origin. Since the origin corresponds to the constant monomial (\textit{i.e.}~degree $0$), $P_s^{(n)}$ corresponds to all monomials of degree less than or equal to $s$.

We now prove a lower bound on the size of projective toric designs. We note that this bound is compatible with the asymptotic bound given in \cref{thm:asymptoticdesigns}. One can see this by using the asymptotic expansion of the binomial coefficients in \cref{eq:Pst-size}.

\begin{proposition}
    \label{prop:minimal-toric-t-design}
    Let $n\in \bbN$ and $(X,\Sigma,\nu)$ be a finite $P(T^n)$ $t$-design. Then $\abs{X} \geq G_{n-1}\pargs{\floor{t/2}}$, where $G_{n-1}(s)$ is given in \cref{eq:Pst-size}.
\end{proposition}
\begin{proof}
    We prove the bound for even $t$. The bound for odd $t$ is then automatically valid since the minimal size of a $(t+1)$-design is at least as large as the minimal size of a $t$-design. We therefore restrict our attention to even $t$ for the rest of the proof.

    Since $X$ is a finite, discrete measure space, we can rewrite $\int_X (\cdot) \dd\nu$ as $\sum_{\phi\in X}v(\phi)(\cdot)$.
    The projective toric $t$-design condition can be expressed as follows.
    Let each $\phi\in X$ label a basis element of $V\coloneqq \bbC^{\abs{X}}$ so that $\set{\ket\phi \mid \phi \in X}$ is an orthonormal basis of $V$. Then for $\mathbf k\in P_{t/2}^{(n)}$, define $\ket{\mathbf k} = \sum_{\phi\in X} \sqrt{v(\phi)} \e^{\i \mathbf k \cdot \phi} \ket\phi$. The $t$-design condition is equivalently stated as $\bra{\mathbf k}\ket{\mathbf k'} = \delta_{\mathbf k,\mathbf k'}$. Hence, $\set*{\ket{\mathbf k} \mid \mathbf k \in P_{t/2}^{(n)}}$ must be orthonormal in $V$, meaning that $\abs*{P_{t/2}^{(n)}} \leq \dim V = \abs{X}$. The proposition then follows from \cref{eq:crystal-ball}.
\end{proof}

Furthermore, we can prove that a minimal $t$-design for even $t$ must be uniformly weighted.

\begin{proposition}
    Let $X \subset P(T^n)$ and let $v\colon X \to (0,\infty)$ define a weighted discrete measure on $X$. Suppose the measure space defined by $X$ and $v$ is a minimal $t$-design with $t$ even. Then $v(\theta) = 1/\abs{X}$.
\end{proposition}
\begin{proof}
    This proof essentially follows that of Ref.~\cite[Thm.~2.2]{coolsMinimalCubatureFormulae1996}. The $P(T^n)$ $t$-design condition is written as $MM^\dag = \bbI_{\abs*{P_{t/2}^{(n)}}\times \abs*{P_{t/2}^{(n)}}}$, where $M_{\mathbf k,\theta} = \sqrt{v(\theta)} \e^{\i \mathbf k \cdot \theta}$. If $X$ is minimal---that is, if $\abs{X} = \abs*{P_{t/2}^{(n)}}$---then $M$ is a square matrix so that $MM^\dag = \bbI$ if and only if $M^\dag M = \bbI$. From the latter condition, it follows that $\delta_{\theta,\theta'} = \sqrt{v(\theta)v(\theta')} \sum_{k\in P_{t/2}^{(n)}} \e^{\i \mathbf k \cdot (\theta-\theta')}$. When $\theta=\theta'$, we therefore find that $v(\theta) = 1/\abs*{P_{t/2}^{(n)}} = 1/\abs{X}$.
\end{proof}

Finally, we conjecture that the bound given in \cref{prop:minimal-toric-t-design} is tight for even $t$.

\begin{conjecture}
    \label{conj:tight-minimal}
    When $t$ is even, the bound given in \cref{prop:minimal-toric-t-design} is tight in the sense that there are infinitely many dimensions $n$ for which the bound is saturable.
\end{conjecture}

In \cref{sec:difference}, we show how minimal $t$-designs are related to difference sets. Using this connection, we construct an infinite family of minimal $2$-designs that indeed saturate the bound given in \cref{prop:minimal-toric-t-design}, and we derive a bound on the size of dense difference sets.

\section{Relation to difference sets}
\label{sec:difference}

We say that $X\subset P(T^n)$ is a \emph{group toric $t$-design} if $X$ is a $t$-design and also inherits group structure from $P(T^n)$. In this section, we consider the case when $X$ is a cyclic group for finite $n$ and a circle group for $n=\infty$. In this case, we find connections to Sidon sets and difference sets \cite{taoAdditiveCombinatorics2006}.
Using this, in \cref{sec:singer-design-family}, we construct minimal \( P(T^n) \) \( 2 \)-designs whenever \( n-1 \) is a prime power, and more generally we construct \( t \)-designs of size \( \frac{(n-1)^{t+1}-1}{n-2} \) whenever \( n-1 \) is a prime power.

We begin with the infinite case. Suppose that $X\subset P(T^\infty)$ is a $t$-design and isomorphic to the circle group $\U(1)$. Then there is a single element $z \in \bbZ^\infty$ such that $X = \{\theta z = (\theta z_1, \theta z_2, \dots) \mid \allowbreak \theta \in [0,2\pi]\}$. In order for $X$ to be a design, it must be that
\begin{equation}
    \int_0^{2\pi} \exp\pargs{\i \theta \sum_{j=1}^t (z_{a_j} - z_{b_j})}   \frac{\dd\theta}{2\pi} = \begin{cases}
        1 & \text{if } \mset{a_j \mid j\in\set{1,\dots, t}}  = \mset{b_j \mid j\in\set{1,\dots, t}} \\
        0 & \text{otherwise}
    \end{cases}
\end{equation}
for all $a,b \in \bbN^t$. It follows that $z$ must satisfy
\begin{equation}\label{eq:difference-set-cond}
    \parentheses{\sum_{j=1}^t z_{a_j} = \sum_{j=1}^t z_{b_j}} \iff \parentheses{\mset{a_j \mid j\in\set{1,\dots, t}}  = \mset{b_j \mid j\in\set{1,\dots, t}}}.
\end{equation}
In other words, the sum of any $t$ elements of $z$ must be unique.
If we restrict $z$ to be in $\bbZ_{\geq 0}^\infty$, then \cref{eq:difference-set-cond} is exactly the definition for $z$ to be a \emph{$B_t$ set} \cite[Def.~4.27]{taoAdditiveCombinatorics2006}.
In the special case of $t=2$, we need to find a $z \in \bbZ_{\geq 0}^\infty$ such that $z_a + z_b = z_c + z_d$ if and only if $\mset{a,b} = \mset{c,d}$. Such a $z$ is called a \emph{Sidon set} \cite{taoAdditiveCombinatorics2006}.

\begin{definition}[\( B_t \) and Sidon sets \cite{taoAdditiveCombinatorics2006}]
    A \emph{\( B_t \) set}\footnote{Note that we are considering $z$ to be a tuple and yet calling it a difference ``set''. It is understood that we are talking about the set $\set{z_a \mid a \in \bbN}$.}
    is an element \( z \in \bbZ_{\geq 0}^\infty \) satisfying \cref{eq:difference-set-cond} for all \( a,b \in \bbN^t \).
    A \emph{Sidon set} is a \( B_2 \) set.
\end{definition}

We have therefore proven the following proposition.

\begin{proposition}
    Group $P(T^\infty)$ $t$-designs isomorphic to the circle group are in one-to-one correspondence with $B_t$ sets.
\end{proposition}

We next give a simple example of a $B_t$ set.

\begin{example}[Exponential $B_t$ set]
    Let $z\in \mathbb{Z}^\infty$ be defined by $z_a = t^a$. In this case, $z_a$ written in base $t$ is $100\dots 0$, a $1$ followed by $a$ $0$s. It follows easily that every sum is unique up to reordering.
\end{example}

We now discuss finite $n$. Suppose that $X\subset P(T^n)$ is a $t$-design and isomorphic to the cyclic group $\bbZ_m$. It follows that $X$ is a size $m$ $t$-design and is generated by a fixed $z \in \bbZ_m^n$. In order for $X$ to be a design, it must be that
\begin{equation}
    \sum_{d=0}^{m-1}\exp\pargs{\frac{2\pi\i d}{m} \sum_{j=1}^t (z_{a_j} - z_{b_j})}
    =
    \begin{cases}
        1 & \text{if } \mset{a_j \mid j\in\set{1,\dots, t}}  = \mset{b_j \mid j\in\set{1,\dots, t}} \\
        0 & \text{otherwise}
    \end{cases}
\end{equation}
for all $a,b \in I_n^t$, where recall that $I_n = \set*{1,2,\dots,n}$. It follows that $z$ must satisfy
\begin{equation}\label{eq:finite-difference-set-cond}
    \parentheses{\sum_{j=1}^t z_{a_j} \equiv \sum_{j=1}^t z_{b_j} \pmod{m}} \iff \parentheses{\mset{a_j \mid j\in\set{1,\dots, t}}  = \mset{b_j \mid j\in\set{1,\dots, t}}}.
\end{equation}
In other words, the sum of any $t$ elements of $z$ must be unique, or equivalently,
\begin{equation}
    \bigg\lvert \bigg\{ \sum_{j=1}^t z_{a_j} \bmod m  \mid a \in I_n^t  \bigg\} \bigg\rvert = \binom{n+t-1}{t} .
\end{equation}
\cref{eq:finite-difference-set-cond} is precisely the definition for $z$ to be a \emph{$B_t$ mod $m$ set of size $n$} \cite{taoAdditiveCombinatorics2006}.

\begin{definition}[Modular \( B_t \) and Sidon sets \cite{taoAdditiveCombinatorics2006}]
    A \emph{\( B_t \) mod \( m \) set of size \( n \)}
    is an element \( z \in \bbZ_m^n \) satisfying \cref{eq:finite-difference-set-cond} for all \( a,b \in I_n^t \).
    A \emph{Sidon set of size \( n \) mod \( m \)} is a \( B_2 \) mod \( m \) set of size \( n \).
\end{definition}

We have therefore shown the following proposition.

\begin{proposition}
    \label{prop:group-designs-Bt-sets}
    Group $P(T^n)$ $t$-designs isomorphic to the cyclic group $\bbZ_m$ are in one-to-one correspondence with $B_t$ mod $m$ sets of size $n$.
\end{proposition}

Given \cref{prop:group-designs-Bt-sets} and the bound in \cref{prop:minimal-toric-t-design}, we immediately arrive at the following corollary.

\begin{corollary}
    \label{cor:Bt-set-bound}
    Any $B_t$ mod $m$ set must have size $n$ satisfying $m \geq G_{n-1}\pargs{\floor{t/2}}$, where $G_{n-1}(s)$ is given in \cref{eq:Pst-size}. Furthermore, if \cref{conj:tight-minimal} is true, then this bound is tight for even $t$ in the sense that there are infinitely many dimensions $n$ for which the bound is saturable.
\end{corollary}

We have been unable to find the bound in \cref{cor:Bt-set-bound} in the existing literature on difference sets. If this bound is indeed new, it illustrates the utility of studying projective toric designs due to the many interesting mathematical objects to which they relate.

In the special case of $t=2$, $B_{t=2}$ mod $m$ sets are called a \emph{Sidon sets of size $n$ mod $m$}. Notably, by a simple counting argument, any Sidon set of size $n$ mod $m$ must satisfy $m \geq n(n-1)+1$.\footnotemark{} Further, for many but not all $n$, this bound can be saturated, as we discuss later. When the bound is saturated, we say the Sidon set is \textit{dense}. Hence, for every $n$ for which there is a Sidon set of size $n$ mod $n(n-1)+1$, there is a \emph{minimal} $P(T^n)$ $2$-design---that is, a $P(T^n)$ $2$-design of size $n(n-1)+1$, hence saturating the lower bound from \cref{prop:minimal-toric-t-design}.
\footnotetext{The Sidon set condition can be restated as stipulating that $z_a - z_c \equiv z_d - z_b$ if and only if $\mset{a,b} = \mset{c,d}$. We therefore need $z_a - z_c$ to be unique for all $a$ and $c$. First choose an $a \in I_n$ and then choose a $c\in I_n$ with $c\neq a$. This gives us $n(n-1)$ distinct values. Further, we have one more value---namely $0$---coming from when $a=c$.}

For one example of a dense Sidon set, consider $n = 6$ and $m=G_{n-1}(1) = n(n-1)+1 = 31$. Then one can easily check that $z = (0,1,3,8,12,18)$ is a Sidon set and thus gives rise to a $P(T^6)$ $2$-design of size $31$. A simple numerical search however reveals that there does not exist a Sidon set of size $7$ mod $7(7-1)+1 = 43$. Furthermore, by the classification of finite abelian groups, any group of order \( 43 \) must be isomorphic to \( \bbZ_{43} \). Therefore, we have the following corollary.

\begin{corollary}
    Either there are no $P(T^7)$ $2$-designs of size saturating the lower bound given in \cref{prop:minimal-toric-t-design}, or such a saturating design cannot be isomorphic to a group.
\end{corollary}

\subsection{Explicit families of designs from Singer sets}
\label{sec:singer-design-family}

There is a general construction of dense Sidon sets---called Singer sets---whenever $n-1$ is a prime power \cite{singer_theorem_1938}. Thus, with this, we have constructed minimal $P(T^n)$ $2$-designs whenever $n-1$ is a prime power, and these designs are isomorphic to the cyclic group $\bbZ_{n(n-1)+1}$. For completeness, we review the Singer set construction in \cref{ap:singer-sets}. However, we note that the details of the Singer set construction are not necessary to understand for our work. Indeed, our results only use that such a construction \emph{exists}.
For reference, we provide code for constructing Singer sets \cite{projtoricdesignsGithub}.

Indeed more generally, we review Singer's construction in \cref{lem:singer-Bt-set} of $B_t$ mod $\frac{(n-1)^{t+1}-1}{n-2}$ sets of size $n$ whenever $n-1$ is a prime power. Using \cref{prop:group-designs-Bt-sets}, we have therefore constructed explicit $P(T^n)$ $t$-designs of size $\frac{(n-1)^{t+1}-1}{n-2}$ whenever $n-1$ is a prime power, and these designs are isomorphic to the cyclic group $\bbZ_{\frac{(n-1)^{t+1}-1}{n-2}}$. Furthermore, since the restriction of a $P(T^m)$ $t$-design to $P(T^n)$ for $n \leq m$ is still a $t$-design, it follows that for all $n$ we have constructed explicit $P(T^n)$ $t$-designs of size $\frac{(m-1)^{t+1}-1}{m-2}$, where $m$ is the smallest integer greater than or equal to $n$ such that $m-1$ is a prime power.

Finally, we recall that a $P(T^n)$ $t$-design is a $T^{n-1}$ $t$-design (see the discussion below \cref{eq:repr-proj-toric}). We also note that a $P(T^n)$ $t$-design be made into a $T^n$ $(2t)$-design by twirling over a $T^1 = S^1$ $(2t)$-design. For example, if a set $\Phi \subset \set*{0}\times T^{n-1}$ is a $P(T^n)$ $t$-design, then the set $\set*{\phi + (\theta,\dots,\theta) \mid \phi\in\Phi, \theta\in \Theta}$ is a $T^n$ $(2t)$-design when $\Theta$ is a $S^1$ $(2t)$-design.
We can see this as follows. Suppose $\sum_{j=1}^n \abs{\alpha_j} \leq 2t$. Then,
\begin{salign}[eq:concatenate-proj]
    \frac{1}{\abs*{\Phi}\cdot \abs*{\Theta}}\sum_{\phi\in\Phi}\sum_{\theta\in \Theta} \e^{\i \sum_{j=1}^n \alpha_j (\phi_j+\theta)}
    &= \int_{T^1} \e^{\i \theta \sum_{j=1}^n \alpha_j} \dd{\mu_1(\theta)} \times \frac{1}{\abs*{\Phi}}\sum_{\phi\in\Phi} \e^{\i \sum_{j=1}^n \alpha_j \phi_j} \\
    &= \delta_{\sum_{j=1}^n \alpha_j = 0} \frac{1}{\abs*{\Phi}}\sum_{\phi\in\Phi} \e^{\i \sum_{j=1}^n \alpha_j \phi_j} \\
    &= \delta_{\sum_{j=1}^n \alpha_j = 0} \int_{P(T^n)} \e^{\i \sum_{j=1}^n \alpha_j \phi_j} \dd\mu_{n-1}(\phi) \label{eq:proj-monomial} \\
    &= \prod_{j=1}^n \delta_{\alpha_j, 0},
\end{salign}
where in \cref{eq:proj-monomial} we used that $\e^{\i \sum_{j=1}^n \alpha_j \phi_j}$ is at most a degree $t$ monomial on $P(T^n)$ when $\sum_{j=1}^n \alpha_j = 0$.
The set $\Theta = \big(\frac{2\pi}{2t+1}\big)\bbZ_{2t+1}$ is a $S^1$ $(2t)$-design. We have therefore constructed explicit $T^n$ $(2t)$-designs of size $(2t+1)\times \frac{(m-1)^{t+1}-1}{m-2}$ for all $n$.

\section{Relation to quantum state designs and MUBs}
\label{sec:quantum}

Projective toric designs are closely connected to complex-projective designs \cite{hoggarTDesignsProjectiveSpaces1982,hoggarParametersTDesignsFPd1984,bannaiTightDesignsCompact1985,wootters1989optimal,renesSymmetricInformationallyComplete2004,klappenecker2005mutually, dankert2005efficient, scott_tight_2006,ambainisQuantumTdesignsTwise2007,roberts2017chaos,Kueng2015,dankertExactApproximateUnitary2009}, continuous-variable (CV) rigged designs \cite{iosueContinuousvariableQuantumState2022}, and complete sets of mutually unbiased bases (MUBs) \cite{durtMutuallyUnbiasedBases2010}. These connections arise by concatenating projective toric and simplex designs in order to generate elements in complex-projective space, which in turn satisfy the design condition. We discuss the connection here.
In \cref{sec:quantum-from-toric}, we set up the connection between projective toric designs and quantum state designs and use it to construct \emph{almost minimial} quantum state \( 2 \)-designs (\emph{ie.}~quantum state \( 2 \)-designs in \( d \) dimensions of size \( d^2+1 \)).
Using this connection, in \cref{sec:mubs-zhu,sec:mubs-group}, we find a close connection between projective toric designs and MUBs, and we use this connection to prove various results.
Namely, in \cref{sec:mubs-zhu}, we disprove Zhu's conjecture, and in \cref{sec:mubs-group}, we characterize a fundamental difference between complete sets of MUBs in prime-power dimensions vs in dimension \( 6 \) in terms of the group structure of the associated projective toric designs.

\subsection{Quantum state designs from projective toric designs}
\label{sec:quantum-from-toric}

Denote the complex unit sphere by $\Omega_d = \set*{z\in \bbC^d \mid \sum_{i=1}^d \abs{z_i}^2 = 1}$, which can be identified with $S^{2d-1}$.
Let $\bbC\bbP^{d-1}$ be complex-projective space $\Omega_d / \U(1)$. Pick an orthonormal basis $\set{\ket n \mid n \in \set*{1,\dots, d}}$ of $\bbC^d$. A polynomial $f$ on $\Omega_d$ descends to a well-defined polynomial on $\bbC\bbP^{d-1}$ if and only if it is invariant under the action of $\U(1)$---that is, $f(\e^{\i\theta} \ket\psi) = f(\ket\psi)$ for all $\theta$ and $\ket\psi\in\Omega_d$. It follows that all degree $t$ monomials on $\bbC\bbP^{d-1}$ are of the form $\prod_{i=1}^t \bra{a_i}\ket\psi\bra\psi\ket{b_i}$ for $a, b \in I_d^t$ (recall that $I_d = \set*{1,2,\dots,d}$). A $\bbC\bbP^{d-1}$ $t$-design is thus a measure space $(X, \Sigma, \nu)$ such that, for all $a,b\in I_d^t$,
\begin{equation}
    \int_X \parentheses{\prod_{i=1}^t \bra{a_i}\ket\psi\bra\psi\ket{b_i}} \dd{\nu(\psi)} = \int_{\bbC\bbP^{d-1}} \parentheses{\prod_{i=1}^t \bra{a_i}\ket\psi\bra\psi\ket{b_i}} \dd\psi = \frac{\Pi_t^{(d)}(a;b)}{\Tr\Pi_t^{(d)}} ,
\end{equation}
where $\Pi_t^{(d)}$ is the projector onto the symmetric subspace of $(\bbC^d)^{\otimes t}$,
\begin{equation}
    \Pi_t^{(d)}(a;b) \coloneqq \parentheses{\bigotimes_{i=1}^t \bra{a_i}} \Pi_t^{(d)} \parentheses{\bigotimes_{i=1}^t \ket{b_i}},
\end{equation}
and $\dd\psi$ denotes the Fubini-Study volume measure on $\bbC\bbP^{d-1}$. The last equality is a simple consequence of Schur's lemma and the unitary invariance of $\dd\psi$ \cite{scott_tight_2006,roberts2017chaos}\cite[Ap.~C3]{iosueContinuousvariableQuantumState2022}.

Let $\Delta^{d-1} = \set*{p\in[0,1]^d \mid \sum_{i=1}^d p_i = 1}$ denote the $(d-1)$-dimensional simplex. Simplex $t$-designs have analogous definitions to those of toric and complex-projective designs \cite{stroud_approximate_1971,hammer_numerical_1956,baladram_on_2018,kuperbergNumericalCubatureArchimedes2004,kuperberg_numerical_ec_2004}. Any vector $\ket\psi \in \Omega_d$ can be represented as $\ket{p, \phi} \coloneqq \sum_{n=1}^d \sqrt{p_n} \e^{\i \phi_n} \ket n$ for some (not necessarily unique) $p\in \Delta^{d-1}$ and $\phi\in T^d$.
For a complex unit vector $\ket\psi\in \Omega_d$, let $[\ket\psi]$ denote the equivalence class corresponding to a point in $\bbC\bbP^{d-1}$.
Let $\pi\colon \Delta^{d-1}\times P(T^d) \to \bbC\bbP^{d-1}$ be defined by $(p, \phi) \mapsto [\ket{p,\phi}]$, where $\phi$ is any representative of an equivalence class in $T^d/T = P(T^d)$.
The pullback of the Fubini-Study volume form along $\pi$ is precisely the Lebesgue measure on $\Delta^{d-1}$ times the Lebesgue measure on $P(T^d)$ (see \cref{ap:fubini-study}). Together, this implies that the concatenation of a $\Delta^{d-1}$ $t$-design and a $P(T^d)$ $t$-design yields a $\bbC\bbP^{d-1}$ $t$-design \cite{kuperbergNumericalCubatureArchimedes2004, iosueContinuousvariableQuantumState2022}.

We note that the analogous result holds for the complex sphere $\Omega_d$; namely, concatenation of a $\Delta^{d-1}$ $t$-design and a toric $(2t)$-design (see \cref{def:toric-design}) yields a $\Omega_d$ $t$-design.
The reason that we only need a projective toric design in the $\bbC\bbP^{d-1}$ case, as opposed to a full toric design as in the $\Omega_d$ case, is because polynomials on $\bbC\bbP^{d-1}$ are more restricted than on $\Omega_d$. On $\Omega_d$, $z_1z_2\bar z_3$ is a valid monomial. On the other hand, this is an invalid monomial on $\bbC\bbP^{d-1} = \Omega_d/\U(1)$ since it varies under the action of $\U(1)$.

One particularly nice simplex \( 2 \)-design contains the extremal points and the centroid (see e.g.~\cite[Thm.~C4]{iosueContinuousvariableQuantumState2022}), which we show in the following proposition (see \cref{ex:almost-minimal-design} for a simple example).

\begin{proposition}
    \label{prop:simplex-design}
    Let \( c = (1/d, \dots, 1/d) \in \Delta^{d-1} \) be the centroid of the simplex and \( e^{(1)} = (1, 0, \dots, 0) \), \( \dots \), \( e^{(d)} = (0, 0, \dots, 1) \) be the extremal points, and define \( D = \set{c, e^{(1)}, \dots, e^{(d)}} \). Define the weight function \( w\colon D \to [0, 1] \) by \( w(e^{(j)}) = \frac{1}{d(d+1)} \) and \( w(c) = \frac{d}{d+1}  \). The discrete probability space defined by \( (D, w) \) is a \( \Delta^{d-1} \) \( 2 \)-design.
\end{proposition}

When concatenating the extremal points \( e^{(j)} \) of the simplex with a projective toric design, we get the basis vectors $[\ket j] \in \bbC\bbP^{d-1}$, since $[\ket{e^{(j)}, \phi}] = [\ket j]$ for any $\phi$. When concatenating the centroid with a finite-sized projective toric design $X$, we get a collection of points $\set{[\ket{c, \phi}] \in \bbC\bbP^{d-1} \mid \phi \in X}$. Hence, the total number of points in the resulting complex-projective design is $d + \abs{X}$. Recalling \cref{prop:minimal-toric-t-design}, we have that $\abs{X} \geq d(d-1)+1$. Furthermore, from \cref{sec:difference}, we found an explicit construction using Singer sets of these minimal projective toric designs whenever $d+1$ is a prime power. It follows that the resulting complex-projective $2$-design is of size $d^2+1$ (note that this complex-projective design is not uniformly-weighted). Interestingly, the smallest possible complex-projective $2$-design---also called a SIC-POVM---has size $d^2$. The existence of SIC-POVM's in all dimensions $d$ is still an open problem.

These \textit{almost-minimal} $\bbC\bbP^{d-1}$ $2$-designs that we just constructed using Singer sets---$\bbC\bbP^{d-1}$ $2$-designs of size \( d^2+1 \)---were first constructed in Ref.~\cite{bodmannAchievingOrthoplexBound2015}. Notably, however, our utilization of projective toric designs indicates a possible path toward extending such constructions to higher $t$-designs.

\begin{example}[$d=3$]
    \label{ex:almost-minimal-design}
    We construct the above almost-minimal $\bbC\bbP^{d-1}$ $2$-design in the case of $d=3$.
    Let us utilize the minimal $P(T^3)$ $2$-design given by the mod $7$ Sidon set $(0, 1, 3)$. The corresponding projective toric design is given by the phases
    \begin{equation}
        \set{
            \parentheses{0, \frac{2\pi k}{7}, \frac{2\pi k}{7} \times 3 } \mid k \in \bbZ_7
        },
    \end{equation}
    where we understand $(0, \theta,\phi) \in T^3$ to be a representative of an equivalence class in $P(T^3)$.
    Denote by $(p_0,p_1,p_2)$ an element of $\Delta^2$. Consider the $\Delta^2$ $2$-design from \cref{prop:simplex-design} given by the centroid $(1/3,1/3,1/3)$ weighted by $3/4$ and the extremal points $(1, 0, 0)$, $(0, 1, 0)$, $(0, 0, 1)$ each weighted by $1/12$.
    Finally,
    denote points in $\bbC\bbP^2$ by $[\ket\psi]$ for $\ket\psi$ a unit vector in $\bbC^3$, and fix an orthonormal basis $\set{\ket0,\ket1,\ket2}$.
    Let
    \begin{equation}
        X = \set{[\ket0], [\ket1], [\ket2]}
        \cup
        \set{[\ket{\psi_k}] \coloneqq\brackets{\frac{1}{\sqrt{3}}\parentheses{\ket0+ \e^{2\pi\i k / 7}\ket1 + \e^{2\pi\i k \times 3 / 7}\ket 2}} \mid k \in \bbZ_7}.
    \end{equation}
    Turn $X$ into a discrete measure space by weighting $[\ket0],[\ket1],$ and $[\ket2]$ each by $1/12$ and each $[\ket{\psi_k}]$ by $3/(4\times 7)$. The resulting measure space is a $\bbC\bbP^2$ $2$-design of size $10$.
\end{example}

Similar to the finite-dimensional case above, we see analogous results in infinite-dimensional quantum systems.
Ref.~\cite{blume-kohoutCuriousNonexistenceGaussian2014} introduced the notion of a continuous variable (CV) $t$-design.
Ref.~\cite{iosueContinuousvariableQuantumState2022} proved that such designs do not exist and therefore introduced \textit{rigged} CV $t$-designs.
A simplex design can be generalized to the unnormalized infinite-dimensional simplex. It then follows that the concatenation of an infinite-dimensional simplex $t$-design and a $P(T^\infty)$ $t$-design yields a rigged CV $t$-design.
We therefore see that designs on the infinite-dimensional projective torus $P(T^\infty)$ are closely related to designs on other infinite-dimensional spaces.

\subsection{MUBs and quantum state designs: counterexamples to Zhu's conjecture}
\label{sec:mubs-zhu}

In this subsection, we explicitly derive the relationship between complete sets of MUBs and projective toric $2$-designs. We then use this relationship to disprove Zhu's conjecture regarding the structure of MUBs \cite[Conj.~1]{zhu2015mutually-unbias}.

The non-existence of complete sets of MUBs in non-prime-power dimensions has been a well-known and long-outstanding question in quantum information theory. Since the question of non-existence has proven incredibly difficult, there have been many related conjectures made in hopes of making some progress. One such conjecture is as follows.

\begin{quote}
    \textbf{Zhu's conjecture}~\cite[Conj.~1]{zhu2015mutually-unbias}.
    Any (uniformly-weighted) quantum state \( 2 \)-design in \( d \) dimensions of size no more than \( d(d+1) \) is either a complete set of MUBs or a SIC-POVM.
\end{quote}

See also Ref.~\cite{avella2024cyclic-measurem} for a discussion on Zhu's conjecture.
In this subsection, we disprove this conjecture by explicitly constructing counterexamples that utilize the difference set construction of projective toric designs. We begin by recalling the definition of complete sets of MUBs. We then show the relationship of MUBs to projective toric designs, and we restate Zhu's conjecture in terms of projective toric designs.
We then show explicit examples of projective toric designs that violate Zhu's conjecture.
We then prove a nice characterization of complete sets of MUBs in dimension \( 6 \) in terms of non-group projective toric \( 2 \)-designs (\emph{cf.}~\cref{prop:nongroup-dim-6}).
Finally, we discuss one possible modification of Zhu's conjecture and discuss potential paths towards proving this new conjecture (\emph{cf.}~\cref{conj:mubs}).

For brevity, we take Ref.~\cite[Thms.~3,~4]{klappenecker2005mutually} as our definition\footnote{The second condition in \cref{def:mubs} is actually implied by the first and third \cite{Matolcsi_2021} (we thank Daniel McNulty for pointing this out).} of a complete set of MUBs.

\begin{definition}
    \label{def:mubs}
    A set \( \calM \subset \bbC\bbP^{d-1} \) is a complete set of MUBs if and only if
    \begin{enumerate}
        \item \( \abs{\calM} = d(d+1) \),
        \item \( \calM \) is a (uniformly-weighted) quantum state \( 2 \)-design (ie.~a \( \bbC\bbP^{d-1} \) \( 2 \)-design), and
        \item for every \( \ket\psi\neq \ket\varphi \in \calM \), \( \abs{\bra\psi\ket\varphi}^2 \in \set{0, 1/d} \).
    \end{enumerate}
\end{definition}

In Ref.~\cite[App.~F]{iosueContinuousvariableQuantumState2022}, it was shown that projective toric designs are closely related to complete sets of MUBs. For completeness, we state this relationship (modified from Ref.~\cite{iosueContinuousvariableQuantumState2022} to use the terminology of our paper)
as a theorem and sketch the proof of the direction that is most important for us in this work.

\begin{theorem}[App.~F~of~Ref.~\cite{iosueContinuousvariableQuantumState2022}]
    \label{thm:mub-is-toric}
    A complete set of MUBs \( \calM \) exists in dimension \( d \) if and only if there exists a uniformly-weighted \( P(T^d) \) \( 2 \)-design \( X \) of size \( \abs{X} = d^2 \) satisfying
    \begin{equation}
        \label{eq:toric-angle}
        \forall \phi \neq \theta \in X\colon~~ \abs{\sum_{j=1}^{d} \e^{\i(\phi_j - \theta_j)}}^2 \in \set{0, d} .
    \end{equation}
\end{theorem}
\begin{proof}[Proof sketch]
    The ``only if'' direction is proven in Ref.~\cite[Lem.~F.2]{iosueContinuousvariableQuantumState2022} (in Ref.~\cite{iosueContinuousvariableQuantumState2022}, a projective toric design is referred to as a torus design).

    We now sketch the proof of the ``if'' direction. Let \( X \) be a uniformly-weighted \( P(T^d) \) \( 2 \)-design. As discussed in \cref{sec:quantum-from-toric}, we can concatenate \( X \) with the simplex \( 2 \)-design given in \cref{prop:simplex-design} to yield the \( \bbC\bbP^{d-1} \) \( 2 \)-design \( \calM \) consisting of the elements
    \begin{enumerate}
        \item \( \ket n \) for \( n \in \set{0,\dots, d-1} \), each weighted by \( \frac{1}{d(d+1)} \); and
        \item \( \ket{c, \phi} \) for \( \phi\in X \), each weighted by \( \frac{d}{d+1}\times \frac{1}{\abs{X}} = \frac{1}{d(d+1)} \),
    \end{enumerate}
    where recall that \( c = (1/d,\dots,1/d) \in \Delta^{d-1} \) is the centroid and \( \ket{p, \phi} = \sum_{n=1}^d \sqrt{p_n} \e^{\i\phi_n} \).
    Hence, \( \calM \) is a uniformly-weighted \( 2 \)-design of size \( \calM = d(d+1) \). To show that \( \calM \) is a complete set of MUBs, the only thing left to show is condition (3) of \cref{def:mubs}. The only nontrivial overlaps to consider are \( \bra{c,\theta}\ket{c,\phi} \) for \( \phi\neq \theta\in X \), giving
    \begin{equation}
        \abs{\bra{c,\theta}\ket{c,\phi}}^2
        = \abs{\frac{1}{d}\sum_{j=1}^d \e^{\i (\phi_j - \theta_j)}}^2
        \in \set{0, 1/d},
    \end{equation}
    which comes by assumption of \cref{eq:toric-angle}.
\end{proof}

Using \cref{thm:mub-is-toric} and the fact that (as we saw in the proof) any uniformly-weighted \( P(T^d) \) \( 2 \)-design yields a uniformly-weighted quantum state \( 2 \)-design of size \( d(d+1) \) via concatenation with the simplex design given in \cref{prop:simplex-design}, we see that we can now rephrase part of Zhu's conjecture as follows.

\begin{quote}
    \textbf{Rephrasing of (part of) Zhu's conjecture}~\cite[Conj.~1]{zhu2015mutually-unbias}.
    Any uniformly-weighted \( P(T^d) \) \( 2 \)-design \( X \) of size \( \abs{X} = d^2 \) must satisfy \cref{eq:toric-angle}.
\end{quote}

We note that we are rephrasing Zhu's conjecture as it pertains to uniformly-weighted quantum state $2$-designs of size $d(d+1)$, as this is the case of interest for complete sets of MUBs and for the remainder of this manuscript. We do not make any statements regarding the existence of uniformly-weighted quantum state $2$-designs of size $\lneq d(d+1)$.

In the language of our paper, Zhu's conjecture can be seen as conjecturing that a complete set of MUBs exists in dimension \( d \) if and only if there exists a uniformly-weighted \( P(T^d) \) \( 2 \)-design \( X \) of size \( \abs{X} = d^2 \).
Given the Sidon set construction of \( P(T^d) \) \( 2 \)-designs from \cref{sec:difference}, we see that we can disprove Zhu's conjecture by finding a Sidon set of size \( d \) mod \( d^2 \) that does not satisfy \cref{eq:toric-angle}.
Via brute force numerical searches for Sidon sets of size \( d \) mod \( d^2 \) for small \( d \), we can find many counterexamples.
Indeed, there are \( 288 \) such Sidon sets when \( d=6 \) (see our code \cite{projtoricdesignsGithub}), hence yielding \( 288 \) counterexamples to Zhu's conjecture in dimension \( 6 \). The simplest counterexample however occurs in dimension \( d=3 \), as we show in the following example.

\begin{example}[Counterexample to Zhu's conjecture in dimension \( 3 \)]
    Utilizing \( \set{0, 1, 3} \), which is a Sidon set of size \( 3 \) mod \( 9 \), we arrive at the following uniformly-weighted quantum state \( 2 \)-design of size exactly \( d(d+1) = 12 \) that is \emph{not} a complete set of MUBs:
    \begin{equation}
        \set{
            \begin{pmatrix}
                1 \\0\\0
            \end{pmatrix},
            \begin{pmatrix}
                0 \\1\\0
            \end{pmatrix},
            \begin{pmatrix}
                0 \\0\\1
            \end{pmatrix}
        } \bigcup
        \set{
            \frac{1}{\sqrt 3}\begin{pmatrix}
                1 \\ \e^{\frac{2\pi\i k}{9}} \\ \e^{\frac{6\pi\i k}{9}}
            \end{pmatrix}
            ~\Biggm\vert~ k \in \set{0, 1, \dots, 8}
        }.
    \end{equation}
\end{example}

\subsection{MUBs and group designs}
\label{sec:mubs-group}

In this subsection, we consider general group projective toric $2$-designs and when they can yield complete sets of MUBs. We find that while such designs can yield complete sets of MUBs is prime-power dimensions, they cannot in dimension in $6$. This illustrates a fundamental difference between the structure of MUBs in prime-power versus non-prime-power dimensions.

Consider the following general parameterization of a subgroup of \( P(T^d) \).
Suppose that the subgroup is isomorphic to \( \bbZ_{\alpha_1} \times \dots \times \bbZ_{\alpha_k} \). We can generate each factor \( \bbZ_{\alpha_j} \) by \( z^{(\alpha_j)} \in \bbZ_{\alpha_j}^d \), where, since we are considering the projective torus, we fix \( z^{(\alpha_j)}_1 = 0 \). We call the subgroup \( X\pargs{\alpha_1,\dots,\alpha_k;z^{(\alpha_1)},\dots,z^{(\alpha_k)}} \), and we have that
\begin{equation}
    X\pargs{\alpha_1,\dots,\alpha_k;z^{(\alpha_1)},\dots,z^{(\alpha_k)}} = \set{\sum_{j=1}^k \frac{2\pi n_j}{\alpha_j}z^{(\alpha_j)} ~\Bigm\vert~ n_1 \in \bbZ_{\alpha_1}, \dots, n_k \in \bbZ_{\alpha_k}}.
\end{equation}
For example, when $d$ is prime, the \( P(T^d) \) $2$-design in the standard MUB construction \cite{wootters1989optimal} (\cref{thm:primesconst}) is
\begin{equation}
    X\pargs{d,d; (0,1,2,3,\dots,d-1), \parentheses{0,1,4,9,\dots,(d-1)^2}}.
\end{equation}
In the case of prime-power dimensions, the construction is a generalization of \cref{thm:primesconst} that uses the field theoretic trace \cite{wootters1989optimal}. Since the trace is linear, one can easily verify the design to be a group.
Thus, for all prime-power \( d\), there are group $2$-designs of size \( d^2 \) satisfying \cref{eq:toric-angle}.

However, via a numerical bruteforce search (see our code \cite{projtoricdesignsGithub}), we find that there are no \( P(T^6) \) group \( 2 \)-designs of size \( 6^2 = 36 \) satisfying \cref{eq:toric-angle}.
The search is done by recognizing that, given the classification of finite abelian groups, there are only four \( \alpha \) to consider:
\begin{salign}
    &\alpha = (\alpha_1, \alpha_2) = (4, 9), \\
    &\alpha = (\alpha_1,\alpha_2,\alpha_3) = (3,3,4), \\
    &\alpha = (\alpha_1,\alpha_2,\alpha_3) = (2,2,9), \\
    &\alpha = (\alpha_1,\alpha_2,\alpha_3,\alpha_4) = (2,2,3,3) .
\end{salign}
We can then explicitly check every generator \( z^{(\alpha_j)} \) for each of these group structures.
We have therefore proven (via exhaustive numerical search) the following proposition, which comes as a corollary of the exhaustive search and \cref{thm:mub-is-toric}.

\begin{proposition}
    \label{prop:nongroup-dim-6}
    If a complete set of MUBs exists in dimension \( 6 \), then there is a uniformly-weighted $P(T^6)$ \( 2 \)-design \( X \) of size \( \abs{X} = 36 \) such that \( X \) is not a subgroup of \( P(T^6) \).
\end{proposition}

Indeed, we believe this proposition highlights why finding complete sets of MUBs in dimension \( 6 \) is difficult: the projective toric designs that form the set must be something more complicated than a group.
We state the generalization of \cref{prop:nongroup-dim-6} to all non-prime-power dimensions as a conjecture.

\begin{conjecture}
    \label{conj:nongroup-mubs}
    Let \( d \) be a non-prime-power. If a complete set of MUBs exists in dimension \( d \), then there is a uniformly-weighted $P(T^d)$ \( 2 \)-design \( X \) of size \( \abs{X} = d^2 \) such that \( X \) is not a subgroup of \( P(T^d) \).
\end{conjecture}

\cref{prop:nongroup-dim-6} and \cref{conj:nongroup-mubs} highlight a fundamental distinction between complete sets of MUBs in prime-power dimensions versus those in non-prime-power dimensions.
To the best of our knowledge, all currently known complete sets of MUBs in prime-power dimensions come from projective toric designs that are groups.
On the contrary, we showed that this is not possible in \( d=6 \), and conjecture it more generally.

Given this result, we make the following conjecture, which can be viewed as a modification of Zhu's conjecture that evades our earlier counterexamples.

\begin{conjecture}
    \label{conj:mubs}
    A complete set of MUBs exists in dimension \( d \) that is not a prime power if and only if there is a uniformly-weighted \( P(T^d) \) \( 2 \)-design \( X \) of size \( \abs{X} = d^2 \) such that \( X \) is not a subgroup of \( P(T^d) \).
\end{conjecture}

Given \cref{conj:mubs}, we naturally conjecture that there do not exist non-group, uniformly-weighted \( P(T^6) \) \( 2 \)-designs of size \( 36 \).
More generally, we conjecture that, in \emph{any} dimension, a complete set of MUBs must come from a group projective toric \( 2 \)-design. We phrase this conjecture in terms of \cref{eq:toric-angle}.

\begin{conjecture}
    Let \( d \) be any dimension and let \( X \) be a \( P(T^d) \) \( 2 \)-design of size \( \abs{X} = d^2 \). If \( X \) is not a subgroup of \( P(T^d) \), then \( X \) does not satisfy \cref{eq:toric-angle}.
\end{conjecture}

If this conjecture is true, then \cref{prop:nongroup-dim-6} would prove the nonexistence of complete sets of MUBs in dimension \( 6 \).


\section{Conclusion and open questions}
\label{sec:conclusion}

In this work, we have developed the theory of projective toric designs and their relation to various other objects in and areas of mathematics and physics. There is still much unknown and we believe there are still many exciting connections to be made. We now discuss various future research directions relating to projective toric designs.

\paragraph{Complete sets of mutually unbiased bases}

In this work, we proved (\emph{cf.}~\cref{prop:nongroup-dim-6}) that the phases in a complete set of MUBs in dimension \( 6 \) must form a uniformly-weighted \( P(T^6) \) \( 2 \)-design of size \( 36 \) that is \emph{not} a subgroup of \( P(T^6) \).
Using projective toric designs that \emph{are} subgroups, we constructed families of quantum state \( 2 \)-designs of size exactly \( d(d+1) \) that are not complete sets of MUBs, thereby disproving Zhu's conjecture \cite[Conj.~1]{zhu2015mutually-unbias}.
Given this result, we proposed a modified verison of Zhu's conjecture (\emph{cf.}~\cref{conj:mubs}) regarding the relationship between complete sets of MUBs and projective toric designs.
An obvious interesting open problem is to prove this conjecture.
Another interesting direction is to try to prove that non-group \( P(T^6) \) \( 2 \)-designs of size \( 36 \) do not exist.
If our conjecture is true, then this result would prove the long-outstanding problem regarding the existence of complete sets of MUBs in dimension \( 6 \).
Finally, projective toric designs may have a close connection to Hadamard matrices, since the latter are also related to MUBs \cite{bengtsson_geometry_2008}.
We leave this interesting question to future work.

\paragraph{Minimal projective toric designs}
In this work, we showed that if $X$ is a $P(T^n)$ $2$-design, then $\abs{X} \geq n(n-1)+1$. Furthermore, using Sidon sets, we showed that the bound can be saturated when $n-1$ is a prime power. However, we also showed that the bound cannot always be satisfied using the Sidon set construction; for example, when $n=7$, the Sidon set construction does not yield a minimal projective toric $2$-design. We thus have the following open question: do projective toric $2$-designs saturating the bound exist for all $n$?

We showed that if the $t$-design is a cyclic group, then the constructions are in one-to-one correspondence with $B_t \bmod \abs{X}$ sets. In the case of e.g.~$n=7$ and $t=2$, $n(n-1)+1=43$ is prime so that the only group design could be a cyclic group. Therefore, if one can prove that a minimal design must be a group, then one would prove that the $t=2$ bound cannot be saturated for all $n$. Must the minimal design be a group?

We further proved that if $X$ is a $P(T^n)$ $t$-design, then $\abs{X} \geq G_{n-1}\pargs{\floor{t/2}}$. We conjectured that the bound is tight when $t$ is even. Can this conjecture be proven? Can the bound be tightened for odd $t$? Can one construct saturating designs?
As we saw in \cref{prop:minimal-toric-t-design}, the lower bound on the size of projective toric \( 2 \)-designs matches the lower bound on the size of dense modular Sidon sets. We believe that the analogous statement holds for all \( t \).
Using the connection between difference sets and projective toric designs, we related dense $B_t \bmod m$ sets to the root lattice $A_{n-1}$ and proved a bound relating the size $n$ of the set and the value of $m$.
This connection seems to be a fruitful area to continue exploring.

\paragraph{Connection to affine/projective planes}

A \textit{finite projective plane} is a tuple $(P, L)$ of a finite set of points $P$ and lines $L\subseteq 2^P$ (where $2^P$ means the power set of set $P$, i.e.\ the set of all subsets of $P$) such that:
\begin{enumerate}
    \item Any two points are elements of a unique common line
    \item Any two lines intersect at a unique point
    \item There exist four points in $P$ such that no line contains more than two of them.
\end{enumerate} Affine planes are defined similarly. A tuple $(P, L)$ can only be a finite projective plane if there exists some $d\in\mathbb{N}$ such that $\vert P\vert=\vert L\vert = d^2+d+1.$ However, finite projective planes have only been constructed for $d$ a prime power, and are known to \textit{not} exist if $d$ is both not the sum of two squares and $d\equiv 1$ or $2$ mod $4.$ These numeric similarities, along with deep connections between combinatorial designs and finite geometry, hint at a deeper connection between projective toric designs and finite projective planes. In addition, projective planes appear in the construction of Sidon sets, and are conjectured to correspond to dense ones \cite{Eberhard_2023}.

Further, a complete set of MUBs yields a finite projective plane, while a SIC-POVM in prime power dimensions yields a finite affine plane \cite{woottersQuantumMeasurementsFinite2004, Saniga_2004,durtMutuallyUnbiasedBases2010}. As mentioned above, MUBs are closely related to projective toric designs, while SIC-POVMs are minimal complex-projective designs. All of this circumstantial evidence begs the question: are there interesting direct connections one can make between projective toric designs and finite planes, either projective or affine?

\paragraph{Connection to other designs}

Recall that complex projective designs can be constructed by concatenating simplex and projective toric designs.
Similarly, rigged continuous variable $t$-designs can be constructed in an analogous way by using $P(T^\infty)$ designs.
One can ask: how much can this result be generalized? Can we use similar constructions for toric varieties and flag varieties? Indeed $\bbC\bbP^n$ is a toric variety with moment map to the associated polytope being the simplex $\Delta^n$. The moment map allows us to project $\bbC\bbP^n$ designs to $\Delta^n$ designs. Projective toric designs allow us to pullback along the moment map and build $\bbC\bbP^n$ designs from $\Delta^n$ designs. How much more general can this result be made?

\paragraph{New families of quantum state designs}
Using the families of $P(T^n)$ $t$-designs constructed in \cref{sec:singer-design-family}, can we generate new interesting families of quantum state $t$-designs? To do this, we need to find families of simplex $t$-designs. In the $t=2$ case, we used a particularly nice simplex $2$-design that allowed us to construct almost-minimal quantum state $2$-designs from minimal projective toric $2$-designs. Can we find similarly nice simplex $t$-designs for $t>2$?
Our construction in \cref{sec:singer-design-family} of \( P(T^d) \) \( t \)-designs of size (asymptotically in \( d \)) \( \approx d^t \) yields quantum state \( t \)-designs of size \( \approx \abs{D}d^t \), where \( D \) is an \( \Delta^{d-1} \) \( t \)-design.
It is an interesting question to study simplex \( t \)-designs to arrive at potentially new explicit constructions of quantum state \( t \)-designs.

\paragraph{Approximate designs}
One can consider approximate projective toric $t$-design, which are points on the projective torus that integrate monomials of degree $\leq t$ up to an error of $\varepsilon$. How does the size of the minimal approximate $t$-design depend on $t$ and $\varepsilon$? If one takes an $\varepsilon_1$-approximate simplex $t$-design and $\varepsilon_2$-approximate projective toric design and concatenates them, what is the $\varepsilon$ with which we get an $\varepsilon$-approximate complex-projective $t$-design? In \cref{ap:approx}, we take the first steps to study such approximate designs. In particular, we define $\varepsilon$-approximate $P(T^n)$ $t$-designs, and we provide an upper bound on the minimum number of points drawn uniformly randomly from $P(T^n)$ needed to form an $\varepsilon$-approximate $P(T^n)$ $t$-design with probability $1-\delta$. We show that the resulting bound depends on the crystal ball sequences of the root lattices $A_{n-1}$ \cite{conway1997lowdimensional-,oeis-crystal-ball} given in \cref{eq:Pst-size}. Using the discussion around \cref{eq:concatenate-proj}, these approximate $P(T^n)$ designs can be lifted to approximate $T^n$ designs, which we recall are then approximate designs on the diagonal subgroup $T(\U(n))$ of $\U(n)$, which are of inherent interest in quantum information theory \cite{haferkamp2023on-the-moments-}. Ref.~\cite{haferkamp2023on-the-moments-} also constructs approximate designs on $T(\U(n))$, which can of course be projected to designs on $P(T^n)$.

\begin{acknowledgments}
    We thank Daniel McNulty, Alexander Barg, Carl Miller, Jakub Czartowski, Karol Życzkowski, Wim van Dam, Greg Kuperberg, Kunal Sharma, Jake Bringewatt, Victor Albert, Shubham Jain, Jonas Haferkamp, and Yu-Xin Wang for helpful discussions.
    We thank two anonymous Referees for various helpful comments and suggestions.
    JTI thanks the Joint Quantum Institute at the University of Maryland for support through a JQI fellowship.
    This work was supported in part by the DoE ASCR Accelerated Research in Quantum Computing program (award No.~DE-SC0020312), DoE ASCR Quantum Testbed Pathfinder program (awards No.~DE-SC0019040 and No.~DE-SC0024220), NSF QLCI (award No.~OMA-2120757),  NSF PFCQC program, AFOSR, ARO MURI, AFOSR MURI, and DARPA SAVaNT ADVENT. Support is also acknowledged from the U.S.~Department of Energy, Office of Science, National Quantum Information Science Research Centers, Quantum Systems Accelerator.
\end{acknowledgments}

\appendix\section{Singer sets}
\label{ap:singer-sets}

In this appendix, we review Singer's construction of Sidon sets of size $p^m+1$ for cyclic groups of size $(p^m)^2+(p^m)+1$ with $p$ a prime \cite[p.~380-381]{singer_theorem_1938} \cite[Sec.~3.5]{obryant_complete_2004} \cite{bose_theorems_1962}.
The existence of these Singer sets implies that there is a $P(T^n)$ $2$-design of size $(n-1)^2+n=n^2-n+1,$ i.e., a minimal one, whenever $n-1$ is prime-power.
More generally, we review Singer's construction of $B_t \pmod {m \coloneqq \frac{(n-1)^{t+1}-1}{n-2}}$ sets (\emph{cf.}~\cref{lem:singer-Bt-set}), which yield $P(T^n)$ $t$-designs of size $m$ whenever $n-1$ is a prime power for any $t$.
We emphasize that everything in this appendix is review.
We also provide code for constructing Singer sets \cite{projtoricdesignsGithub}.

Let $\theta$ be the generator of $\mathbb{F}_{(n-1)^{t+1}}^\times,$ and then let
\begin{equation}
    T_t:=\{0\}\cup \{a \in [(n-1)^{t+1}-1] \mid (\theta^a - \theta) \in \mathbb{F}_{n-1}\subset \mathbb{F}_{(n-1)^{t+1}}\}.
\end{equation}
The inclusion $\mathbb{F}_{n-1}\hookrightarrow \mathbb{F}_{(n-1)^{t+1}}$ is done by identifying the generator of $\mathbb{F}_{(n-1)}^\times$ with $\theta^{\frac{(n-1)^{t+1}-1}{n-2}},$ which makes sense as for any finite field $\mathbb{F}_q,$ $\vert \mathbb{F}^\times_q\vert = q-1,$ and $\mathbb{F}^\times_q$ is cyclic.

Further, note that $\mathbb{F}_{(n-1)^{t+1}}$ is a $(t+1)$-dimensional $\mathbb{F}_{n-1}$-vector space. Thus, $\{\theta^b\}_{b=0}^t$ is a $\mathbb{F}_{n-1}$-basis of $\mathbb{F}_{(n-1)^{t+1}}$. This means that all $\theta^a= \sum_{i=0}^t k_i\theta^i$ for some unique $k_i\in\mathbb{F}_{(n-1)}.$ However, if $\frac{(n-1)^{t+1}-1}{n-2}\vert a,$ we know all $i\geq 1$ have $k_i=0.$

Then, let
\begin{equation}
    S_t((n-1), \theta):= \left\{ l \in \mathbb{Z}_{\frac{(n-1)^{t+1}-1}{n-2}} \mid  l \equiv a \bmod \left(\frac{(n-1)^{t+1}-1}{n-2}\right), a\in T_t \right\}
\end{equation}
be the residues of $T_t$ mod $\frac{(n-1)^{t+1}-1}{n-2}.$ We now recount proofs of some of $S_t((n-1),\theta)$'s properties.

\begin{lemma}
    $\vert S_t((n-1),\theta)\vert = n$.
\end{lemma}
\begin{proof}
    First we note there are $n$ distinct elements of $\mathbb{F}_{(n-1)^{t+1}}$ of the form $\theta + \gamma_a,$ $\gamma_a\in \mathbb{F}_{n-1}$ by the $\mathbb{F}_{n-1}$-linear independence of $\theta$ and $1.$ As all elements of $\mathbb{F}_{(n-1)^{t+1}}$ equal $\theta^a$ for some unique $a\in[(n-1)^{t+1}-1],$ we see that $\vert T_t\vert=n.$ Now, we must show that every element of $T_t$ has a different residue modulo $\frac{(n-1)^{t+1}-1}{n-2}.$

    Suppose $a,a':=a+k\frac{(n-1)^{t+1}-1}{n-2}\in T,$ $k\in \mathbb{Z}_{>0}.$ Then $r:=\theta^{a'}/\theta^a =\theta^{k\frac{(n-1)^{t+1}-1}{n-2}} \in\mathbb{F}_{n-1}.$ But by definition of $T_t,$ $\theta^a = \theta + \gamma_a,$ $\theta^{a'}=\theta + \gamma_{a'}.$ But
    \begin{equation}
        \theta^{a'}=r\theta^a = r\theta + r\gamma_a.
    \end{equation}
    Thus, $r=1,$ meaning $ (n-2)\vert k,$ which means that only $a$ can be in $[(n-1)^{t+1}-1],$ and thus that no two elements of $T_t$ can have the same residue modulo $\frac{(n-1)^{t-1}+1}{n-2}.$
\end{proof}

\begin{lemma}
    \label{lem:singer-Bt-set}
    $S_t((n-1),\theta)$ is a $B_t\ (\mathrm{mod}\ \frac{(n-1)^{t+1}-1}{n-2})$ set.
\end{lemma}
\begin{proof}
    Recall that $\{\theta^i\}_{i=0}^{t}$ is a $\mathbb{F}_{n-1}$-basis of $\mathbb{F}_{(n-1)^{t+1}}.$ In other words, there exist no non-elementwise-zero tuples $(c_i)_{i=0}^t\in \mathbb{F}_{n-1}^{t+1}$ such that
    \begin{equation}
        \sum_{i=0}^t c_i\theta^i=0.
    \end{equation}
    Equivalently, $\theta$ cannot be the root of any polynomial of degree $\leq t$ with $\mathbb{F}_{n-1}$-coefficients.

    Now, consider two multisets $ A,$ $B,$ $\vert A\vert=\vert B\vert \leq t,$ taking entries from $S_t((n-1),\theta).$ Then, by the definition of $S_t((n-1),\theta)$ and $T_t,$ we see that for all $a\in A\cup B$
    \begin{equation}
        \theta^a = \alpha_a (\theta + \gamma_a)
    \end{equation}
    for some $\alpha_a\in\mathbb{F}_{n-1}.$ Now, consider $\Pi_A:=\prod_{a\in A} \theta^a$ and $\Pi_B:=\prod_{b\in B}\theta^b.$ It is clear that $\Pi_B/\Pi_A\in \mathbb{F}_{n-1}$ and only if
    \begin{equation}\label{eq:Btcondition}\sum_{a\in A}a\equiv \sum_{b\in B}b \mod\frac{(n-1)^{t+1}-1}{n-2}.\end{equation}
    Thus, $\Pi_A - \beta_{A,B} \Pi_B=0$ for some $\beta_{A,B}\in\mathbb{F}_{n-1}$ if and only if  Eq.~(\ref{eq:Btcondition}) holds. However, for any $\beta \in \mathbb{F}_{n-1},$ we see that $\Pi_A-\beta \Pi_B$ is a degree-$t$ polynomial equation in $\theta$ with $\mathbb{F}_{n-1}$ coefficients, meaning it cannot have any solutions, meaning the $B_t\ (\mathrm{mod}\ \frac{(n-1)^{t+1}-1}{n-2})$ condition is satisfied.
\end{proof}

\subsection{Explicit example of dense modular Sidon set}

In this appendix, we work through an explicit example of the construction of the Sidon set $S_{t=2}((n-1), \theta)$ for $n = 5 = 2^2 + 1$. We begin by constructing $T_t$. Consider the field $\bbF_{(n-1)^{t+1}} = \bbF_{4^3} = \bbF_{2^6}$. With the irreducible polynomial $f(x) = 1+x^5+x^6 \in \bbF_{2}[x]$, we work in the polynomial representation $\bbF_{2^6} \cong \bbF_2[x] / (f(x))$.

One can check that the generator $\theta$ of the multiplicative group $\bbF_{2^6}^\times$ is $x$ in this representation---in other words, $\abs{\set{x^m \bmod f(x) \mid m \in \bbZ_{63}}} = 63$. We identify $\bbF_{n-1} = \bbF_{2^2} \subset \bbF_{2^6}$ via generating $\bbF_{2^2}^\times$ with
\begin{equation}
    y = x^{\frac{(n-1)^{t+1}-1}{n-2}} = x^{21},
\end{equation}
so that $\bbF_{2^2} = \set{0}\cup\set{y^k \mid k\in \bbZ_{3}}$. Then
\begin{equation}
    T_{t=2} = \set{0} \cup \set{a \in \bbZ_{4^3-1} \setminus \set{0} \mid (x^a - x)\!\! \pmod{f(x)} \in \bbF_{2^2}}.
\end{equation}
Clearly, $1 \in T_{t=2}$. With that out of the way, we can rephrase this as
\begin{equation}
    T_{t=2} = \set{0, 1} \cup \set{a \in \bbZ_{4^3-1} \setminus \set{0,1} \mid \exists k \in \bbZ_3\colon~ x^a - x \equiv y^k \pmod{f(x)}}.
\end{equation}
One can straightforwardly numerically verify that $T_{2} = \set{0, 1, 14, 25, 58}$. To ensure understanding of the construction, we work through why $14 \in T_{2}$. We need to show that $x^{14} - x \equiv y^k \pmod{f(x)}$ for $k=0,1$ or $2$. It turns out that $k=2$ satisfies this equation. In particular,
\begin{equation}
    (x^{14} - x) \!\!\pmod{f(x)} = x^3+x^4+x^5 = y^2 \!\!\pmod{f(x)} = x^{42}\!\! \pmod{f(x)},
\end{equation}
where recall we're working with polynomials over the field $\bbF_2$.
Similarly, for $25$,
\begin{equation}
    (x^{25} - x) \!\!\pmod{f(x)} = 1+x^3+x^4+x^5 = y^1 \!\!\pmod{f(x)} = x^{21}\!\! \pmod{f(x)},
\end{equation}
and for $58$,
\begin{equation}
    (x^{58} - x) \!\!\pmod{f(x)} = 1 = y^0 \!\!\pmod{f(x)}.
\end{equation}
Hence, we have found that $T_2 = \set{0, 1, 14,25,58}$. To get our Sidon set, we compute the residues $S_2 = T_2 \bmod \frac{(n-1)^{t+1}-1}{n-2} = T_2 \bmod 21$, giving
\begin{equation}
    S_2 = \set{0,1,14,4,16} = \set{0,1,4,14,16}.
\end{equation}
One can easily confirm that this is a Sidon set mod $21$. In particular, the set of all sums $a+b \bmod 21$ for $a,b \in S_2$ is $\set{0, 1, 2, 4, 5, 7, 8, 9, 11, 14, 15, 16, 17, 18, 20}$, which has size $15 = \binom{n+t-1}{t} = \binom{6}{2}$, which is the maximal possible size.

\section{Pullback of the Fubini-Study volume form}
\label{ap:fubini-study}

It is shown in Ref.~\cite[Sec.~4.5,~4.7,~7.6]{bengtsson_geometry_2008} that the volume measure on complex projective space is the product of the flat measure on the simplex and the flat measure on the torus. For completeness, in this appendix, we show the same result via a different method.

Let $[Z_0 : \dots : Z_n]$ be homogeneous coordinates on $\bbC\bbP^n$. Consider the coordinate patches $C_0,\dots,C_n$ on $\bbC\bbP^n$, where $C_i = \set{[Z_0 : \dots : Z_n] \mid Z_i \neq 0}$. The volume of $\bbC\bbP^{d-1} \setminus C_0$ is zero, and therefore for the purposes of volume integration we can restrict our attention to $C_0$. On $C_0$, we use the coordinates $z_i \coloneqq Z_i / Z_0$ for $i=1,\dots,n$. The (unnormalized) Fubini-Study volume form $\omega$ can then be written as
\begin{equation}
    \omega =\frac{1}{(1+\sum_{i=1}^n \lvert z_i\rvert^2)^{n+1}} \dd z_1 \wedge \dd\bar z_1 \wedge \dots \dd z_n \wedge \dd\bar z_n.
\end{equation}
We can write $Z_i = \sqrt{p_i}\e^{\i\phi_i}$ for $i=0,\dots, n$ and $\sum_{i=0}^n p_i = 1$. In other words, $p$ is a point on the simplex $p\in \Delta^n \coloneqq \set*{p\in [0,1]^n \mid \sum_i p_i \leq 1}$ (with $p_0 \coloneqq 1-\sum_{i=1}^n p_i$) and $\phi$ is a point on the projective torus $\phi\in P(T^{n+1})$ (\textit{e.g.}~we can choose a representative with $\phi_0=0$). Therefore, $z_i = \sqrt{\frac{p_i}{p_0}} \e^{\i\phi_i-\i\phi_0}$.

Consider the map $\pi\colon \tilde\Delta^n \times P(T^{n+1}) \to C_0$, where $\tilde\Delta^n$ is all $p \in \Delta^n$ satisfying $p_0 > 0$. The map is $\pi^i(p, \phi) = \sqrt{\frac{p_i}{p_0}} \e^{\i\phi_i-\i\phi_0}$.

\begin{proposition}
    The pullback $\pi^\ast \omega$ is
    \begin{equation}
        \pi^\ast \omega = (-1)^{n/2} \dd{p_1} \wedge \dots \dd{p_n} \wedge \dd{\phi_1} \wedge \dots \dd{\phi_n}.
    \end{equation}
\end{proposition}

It follows from this proposition that the unit-volume normalized volume measure on $\bbC\bbP^n$ is equal to the product of the Lebesgue measure on the simplex $\Delta^n$ and the Lebesgue measure on $P(T^{n+1})$ (where recall the latter is equal to the Lebesgue measure on $T^n$).

\begin{proof}[Proof of the proposition]
    We can without loss of generality fix $\phi_0=0$. We can rewrite
    \begin{equation}
        \omega = p_0^{n+1} \dd z_1 \wedge \dd\bar z_1 \wedge \dots \dd z_n \wedge \dd\bar z_n.
    \end{equation}
    Therefore,
    \begin{equation}
        \pi^\ast \omega = p_0^{n+1} \det(J) \dd{p_1} \wedge \dots \dd{p_n} \wedge \dd{\phi_1} \wedge \dots \dd{\phi_n},
    \end{equation}
    where
    \begin{equation}
        J = \begin{pmatrix}A&B\\C&D \end{pmatrix}
    \end{equation}
    is the Jacobian with
    \begin{equation}
        A_{ij} = \frac{\partial \pi^i}{\partial \phi_j}, \quad B_{ij} = \frac{\partial \pi^i}{\partial p_j}, \quad C_{ij} = \frac{\partial \bar\pi^i}{\partial \phi_j}, \quad D_{ij} = \frac{\partial \bar\pi^i}{\partial p_j}.
    \end{equation}
    We can check that
    \begin{equation}
        \frac{\partial \pi^i}{\partial p_j} = \frac{1}{2}\pi^i(p,\phi) \parentheses{\frac{\delta_{ij}}{p_i} + \frac{1}{p_0} }, \qquad \frac{\partial \pi^i}{\partial \phi_j} = \i\delta_{ij}\pi^i(p,\phi).
    \end{equation}
    Therefore, $A$ and $C$ are diagonal and thus commute, meaning that $\det(J) = \det(AD-CB)$. The matrix elements are $(AD-CB)_{ij} = \frac{\i}{p_0}\parentheses{\delta_{ij} + \frac{p_i}{p_0}}$.

    By the matrix determinant lemma \cite{harvilleMatrixAlgebraStatistician1997}, $\det(M + u v^T) = \det(M)(1+v^T M^{-1} u) $ with $M_{ij} = \frac{\i}{p_0}\delta_{ij}$ and $u_i = \i/p_0$ and $v_i = p_i/p_0$, we find that
    \begin{equation}
        \det(J) = \parentheses{\frac{\i}{p_0}}^{n}\parentheses{1+\sum_{i=1}^n \frac{p_i}{p_0}}
        = \parentheses{\frac{i}{p_0}}^{n} \frac{1}{p_0}
        = \frac{(-1)^{n/2}}{p_0^{n+1}}.
    \end{equation}
    The proposition follows.
\end{proof}

\section{Approximate projective toric designs}
\label{ap:approx}

Throughout this appendix, we are concerned with uniform finite $P(T^n)$ $t$-designs; that is, $P(T^n)$ $t$-designs $X$ that are finite and the measure space $(X,\Sigma = \calP(X), \nu)$ is such that $\nu(A) = \abs{A} / \abs{X}$. We restrict to finite $n$. We define approximate projective toric designs and prove a loose bound on the number $M(t,\varepsilon, \delta)$ of uniformly random points needed to form such a design.

For $\mathbf{p}\in\bbN_0^n$, let $f_{\mathbf{p}}(\phi)$ denote the monomial $\prod_{i=1}^n \e^{i \phi_i p_i}$. Notice that $\bar f_{\mathbf{p}}(\phi) = f_{\mathbf{p}}(-\phi) = f_{-\mathbf{p}}(\phi)$.

\begin{definition}
    We say that $C \subset P(T^n)$ is a (uniform) $\varepsilon$-approximate projective toric $t$-design if, for all $\mathbf{p} \in P_t^{(n)}$,
    \begin{equation}\label{eq:torus-cond}
        \abs{\frac{1}{\abs{C}} \sum_{\phi \in C} f_{\mathbf{p}}(\phi) - \int_{P(T^n)} f_{\mathbf{p}} \dd\mu_{n-1}} = \abs{\frac{1}{\abs{C}} \sum_{\phi \in C} f_{\mathbf{p}}(\phi) - \delta_{\mathbf{p},\mathbf 0}} \leq \varepsilon.
    \end{equation}
    Here, $P_t^{(n)}$ is the set defined in \cref{eq:Pst},
    \begin{equation}
        P_t^{(n)} \coloneqq \set{ \mathbf{q} - \mathbf{r} ~\bigg\vert~ \mathbf{q}, \mathbf{r} \in \bbN_0^n, ~  \sum_{i=1}^n q_i = \sum_{i=1}^n r_i = t }.
    \end{equation}
\end{definition}

Note of course that with $\varepsilon = 0$ we recover the definition of an (exact) projective toric design.
There is redundancy in $P_t^{(n)}$. Indeed, if \cref{eq:torus-cond} is satisfied for $\mathbf{p}$, then it is automatically satisfied for $-\mathbf{p}$. Furthermore, \cref{eq:torus-cond} is trivially satisfied for any $C$ when $\mathbf{p} = \mathbf 0$. Hence, we are in fact interested in the set $S_t^{(n)}$ defined by
$S_t^{(n)} \coloneqq (P_t^{(n)} \setminus \set*{\mathbf 0}) / \bbZ_2$,
where $\bbZ_2$ denotes the group action $\mathbf{p} \mapsto \pm \mathbf{p}$.
Therefore,
$C \subset P(T^n)$ is an $\varepsilon$-approximate $t$-design if and only if, for all $\mathbf{p} \in S_t^{(n)}$,
\begin{equation}
    \abs{\frac{1}{\abs{C}} \sum_{\phi \in C} f_{\mathbf{p}}(\phi)} \leq \varepsilon.
\end{equation}

Define the probability space $\calC_M$ to be the ensemble over subsets of $P(T^n)$ of size $M$.
Specifically, to draw a random $C \subset P(T^n)$ from $\calC_M$, we simply draw $M$ uniformly random points from $P(T^n)$ with respect to the Haar measure. We often denote a subset $C \subset P(T^n)$ of size $M$ by $C = \{\phi^{(1)},\dots,\phi^{(M)}\}$, where each $\phi^{(i)} \in P(T^n)$.

\begin{definition}
    Let $M(t,\varepsilon,\delta)$ denote the \textbf{minimum}  $M$ such that a random $C$ drawn from $\calC_M$ is an $\varepsilon$-approximate $P(T^n)$ $t$-design with probability $1-\delta$. In other words,
    \begin{equation}
        \begin{aligned}
            M(t,\varepsilon,\delta) & = \min_{M \in \bbN} M                                                                                               \\
                                    & \text{s.t.~} \Pr_{C \in \calC_M} \bargs{C \text{ is an $\varepsilon$-approx $t$-design on $P(T^n)$}} \geq 1-\delta.
        \end{aligned}
    \end{equation}
\end{definition}

In the following, we find an upper bound on $M(t,\varepsilon,\delta)$.
This tells us that
for any $M \geq M(t,\varepsilon,\delta)$, $C \in \calC_M$ is an $\varepsilon$-approximate $t$-design with probability $\geq 1-\delta$.

\begin{theorem}
    $M(t,\varepsilon,\delta) \leq \frac{G_{n-1}(t) - 1}{2\delta \varepsilon^2}$, where $G_{n-1}(t)$ given in \cref{eq:Pst-size}.
\end{theorem}
\begin{proof}
    Recall that $\abs*{P_t^{(n)}} = 2\abs*{S_t^{(n)}} + 1$, and from \cref{eq:crystal-ball} $\abs*{P_t^{(n)}} = G_{n-1}(t)$. We will therefore prove that $M(t,\varepsilon,\delta) \leq \frac{\abs*{S_t^{(n)}}}{\delta \varepsilon^2}$
    Define the following notation:
    \begin{salign}
        & \Expval_{\phi \in P(T^n)} f(\phi) = \int_{P(T^n)} f \dd\mu_{n-1}\\
        & \Expval_{\phi \in C} f(\phi) = \frac{1}{\abs{C}} \sum_{\phi \in C}f(\phi)\\
        & \Expval_{C \in \calC_M} = \Expval_{\{\phi^{(1)},\dots,\phi^{(M)}\} \in P(T^n)^M}.
    \end{salign}
    For $\mathbf{p}\in S_t^{(n)}$, define
    \begin{equation}
        \Delta(C, \mathbf{p}) \coloneqq \abs{\Expval_{\phi\in C} f_{\mathbf{p}}(\phi) - \Expval_{\phi\in P(T^n)}f_{\mathbf{p}}(\phi)}^2 = \abs{\Expval_{\phi\in C} f_{\mathbf{p}}(\phi)}^2 .
    \end{equation}
    We compute the mean,
    \begin{salign}
        \Expval_{C \in \calC_M} \Delta(C, \mathbf{p})
        &= \frac{1}{M^2}\Expval_{\set*{\phi^{(1)}, \dots, \phi^{(M)}} \in P(T^n)^M} \sum_{i,j=1}^M f_{\mathbf{p}}(\phi^{(i)})\bar f_{\mathbf{p}}(\phi^{(j)})\\
        &= \frac{1}{M^2} \brackets{M \Expval_{\phi\in P(T^n)} f_{\mathbf{p}}(\phi) \bar f_{\mathbf{p}}(\phi) + M(M-1) \parentheses{\Expval_{\phi\in P(T^n)} f_{\mathbf{p}}(\phi)}\parentheses{\Expval_{\theta\in P(T^n)} \bar f_{\mathbf{p}}(\theta)} } \\
        &= \frac{1}{M}.
    \end{salign}
    Meanwhile, we have that
    \begin{salign}
        \Pr_{C \in \calC_M} \big[ C & \text{ is an $\varepsilon$-approx $t$-design on $P(T^n)$} \big]\\
        &= \Pr_{C \in \calC_M} \bargs{\forall \mathbf{p} \in S_t^{(n)}: \Delta(C, \mathbf{p}) \leq \varepsilon^2}\\
        &= 1-\Pr_{C \in \calC_M} \bargs{\exists \mathbf{p} \in S_t^{(n)}: \Delta(C, \mathbf{p}) > \varepsilon^2}\\
        (\text{\em union bound}) &\geq 1-\sum_{\mathbf{p}\in S_t^{(n)}}\Pr_{C \in \calC_M} \bargs{\Delta(C, \mathbf{p}) > \varepsilon^2}\\
        (\text{\em Markov's inequality}) &\geq 1 - \sum_{\mathbf{p} \in S_t^{(n)}} \frac{\Expval_{C\in\calC_M}\Delta(C, \mathbf{p})}{\varepsilon^2}  \\
        &= 1 - \frac{\abs*{S_t^{(n)}}}{M \varepsilon^2}.
    \end{salign}
    Thus, we require that
    \begin{equation}
        1-\frac{1}{M \varepsilon^2} \abs*{S_t^{(n)}} \geq 1-\delta.
    \end{equation}
    It follows that any $M \geq \frac{\abs*{S_t^{(n)}}}{\delta \varepsilon^2}$ satisfies, so that $M(t,\varepsilon,\delta) \leq \frac{\abs*{S_t^{(n)}}}{\delta \varepsilon^2}$.
\end{proof}

\bibliographystyle{quantum}
\bibliography{references}

\end{document}